\documentclass[lettersize,journal]{IEEEtran}
\usepackage{amsmath,amsfonts}
\usepackage{array}
\usepackage{textcomp}
\usepackage{stfloats}
\usepackage{subfigure}
\usepackage{url}
\usepackage{cite}
\usepackage{mathtools}
\usepackage{verbatim}
\usepackage{graphicx}
\usepackage{subfigure}
\usepackage{booktabs}
\usepackage{cite}
\usepackage{tikz}
\usepackage{amsmath,amsthm}
\newtheorem{lemma}{Lemma}  
\usepackage{xcolor}

\newtheorem{theorem}{Theorem}

\usepackage{algorithm}
\usepackage{algpseudocode}

\usepackage[colorlinks=false, pdfborder={0 0 0}]{hyperref}

\definecolor{lime}{HTML}{A6CE39}
\DeclareRobustCommand{\orcidicon}{
\begin{tikzpicture}
\draw[lime, fill=lime] (0,0)
circle[radius=0.16]
node[white]{{\fontfamily{qag}\selectfont \tiny \.{I}D}}; 
\end{tikzpicture}
\hspace{-2mm}
}
\foreach \x in {A, ..., Z}{%
\expandafter\xdef\csname orcid\x\endcsname{\noexpand\href{https://orcid.org/\csname orcidauthor\x\endcsname}{\noexpand\orcidicon}}
}

\graphicspath{{Images00/}}

\begin{document}
\title{Equivalent Flux Compensation for SPMSM Sensorless Control under Parameter Mismatch}
\author{ 

\vskip 1em
	
Fobao Zhou, \emph{Student Member, IEEE}, Jiaqiao Liang, Zhenxiao Yin, \emph{Student Member, IEEE}, \\ Xueyan Wang,
   Yang Shen, \emph{Student Member, IEEE}, Zhongyu Shi, and Hang Zhao$^\ast$\hspace{-1.5mm}\orcidF{}, \emph{Member, IEEE}

\thanks{

    This work is supported by the National Natural Science Foundation of China (No. 52407066), Guangdong Basic and Applied Basic Research Foundation (No. 2024A1515010882) and Guangzhou Municipal Education Bureau 2024 Yangcheng Scholar Project (No. 2024312044).

    Fobao Zhou, Zhenxiao Yin, Xueyan Wang, Yang Shen and Hang Zhao are with the Robotics and Autonomous Systems Thrust at The Hong Kong University of Science and Technology (Guangzhou), Guangzhou 511453,China. 

    Jiaqiao Liang is with Shien-Ming Wu School of Intelligent Engineering at South China University of Technology, Guangzhou 511442, China.

    Zhongyu Shi is with the Department of Engineering Science, University of Oxford, Oxford OX1 2JD, U.K. 
    
	Hang Zhao is the corresponding author (hangzhao@hkust-gz.edu.cn).

}
}
\markboth{Journal of \LaTeX\ Class Files}%
{Shell \MakeLowercase{\textit{et al.}}: A Sample Article Using IEEEtran.cls for IEEE Journals}


\maketitle

\begin{abstract}

Parameter mismatch is the main source of rotor position estimation error in sensorless control of surface-permanent magnet synchronous motors (SPMSMs). To this end, this paper proposes a simple yet efficient equivalent flux compensation (EFC) method that directly estimates the equivalent flux disturbance caused by parameter mismatches in real time. First, the equivalent flux disturbance caused by parameter mismatches is derived from a nonlinear flux observer. Second, a flux update law is proposed to minimize both magnitude and directional errors by leveraging geometric error together with the derived equivalent flux disturbance. To enhance numerical stability, a saturation function is introduced to improve gradient continuity in the update process. Additionally, Lyapunov analysis is employed to ensure the stability of the proposed update law, from which the corresponding error bounds and convergence properties are derived. \textcolor{black}{Finally, experimental results validate that the proposed method fully compensates for the steady-state effects of resistance and flux mismatches, and partially mitigates the influence of inductance variation, effectively constraining the position estimation error within a relatively small range.}

\end{abstract}


\section{Introduction}

\IEEEPARstart{P}{ermanent} Magnet Synchronous Motors (PMSMs) are widely used in electric vehicles, industrial production, and renewable energy systems due to their high efficiency, high power density, and excellent dynamic performance \cite{liu2016research}. Accurate rotor position information is crucial for high-performance vector control. However, mechanical sensors like encoders and resolvers increase cost and size while reducing reliability in harsh environments. Therefore, sensorless control has become a key research focus for PMSMs, offering clear advantages in cost, reliability, and environmental robustness \cite{wang2019position}.

In recent years, various technical approaches have been developed in the field of sensorless control for PMSMs, primarily including back electromotive force (back-EMF) methods \cite{bolognani2014design, he2023optimization, zhang2023commutation}, model reference adaptive systems (MRAS) \cite{prabhakaran2020electromagnetic, yan2022mras}, extended Kalman filters (EKFs) \cite{verrelli2019speed, xiang2025sensorless}, sliding mode observers (SMOs) \cite{wu2024sensorless, yang2024rotor}, and nonlinear flux observers \cite{lee2009sensorless, 11099531, ortega2010estimation, yin2026nonlinear, khlaief2011nonlinear}. Among these, nonlinear flux observers are widely used in practical applications owing to their simple structure and fast dynamic response.

However, the practical implementation of the aforementioned sensorless control methods commonly faces a fundamental challenge: parameter mismatch \cite{xu2018improved, zhou2023robust, bernard2020estimation}. The performance of various observers strongly relies on the accuracy of motor parameters, particularly the stator resistance \cite{hinkkanen2011combined}, inductance \cite{li2015position}, and flux \cite{li2018sensorless}. In real-world operation, factors such as temperature rise, magnetic saturation, and aging effects cause continuous drift in motor parameters, resulting in inherent deviations between controller nominal values and actual parameters \cite{bernard2018convergence, islam2013sensorless}. These deviations directly lead to observation errors, which in turn cause inaccuracies in rotor position estimation and may even induce system instability. Consequently, parameter mismatch has become a critical bottleneck restricting further improvements in sensorless control performance, and is one of the primary sources of rotor position estimation error \cite{choi2016robust, li2025sensorless}.

To address the issue of parameter mismatch in sensorless control, existing approaches can be broadly classified into two categories: parameter identification methods and robust observer methods. Parameter identification methods aim to mitigate the impact of parameter mismatch on the system through online parameter estimation and dynamic compensation \cite{kivanc2018sensorless}. In \cite{hamida2012adaptive}, a sensorless control strategy incorporating an auxiliary filtering system and a gain matrix based on an adaptive interconnected observer was proposed. This method enables simultaneous estimation of rotor speed, rotor position, load torque, stator inductance, and resistance. In \cite{liu2022second}, a fault-tolerant sensorless strategy based on a second-order extended state observer (ESO) and a dual-MRAS observer was proposed. This approach enables the estimation of stator resistance, inductance, and flux using only a single current sensor, thereby enhancing system reliability. Although these methods are effective, they have practical limitations. Simultaneous identification of multiple parameters can significantly increase computational complexity, while strong coupling between identification and state estimation, as well as convergence issues, poses significant challenges \cite{hamida2012adaptive}.

Another approach, namely robust observer techniques, enhances system robustness by incorporating adaptive laws or robust control structures, thereby reducing sensitivity to parameter variations \cite{nguyen2013modeling, zhou2023robust}. In \cite{wang2024eso}, a robust sensorless control method that integrates an ESO with a hybrid flux observer was proposed to mitigate the adverse impact of inductance and flux errors on position estimation accuracy. Similarly, in \cite{woldegiorgis2022sensorless}, an approach based on an extended SMO and a back-EMF compensator is introduced, which effectively suppresses the negative effects of stator resistance and flux variations, thereby significantly improving the robustness of the overall control scheme. \textcolor{black}{Furthermore, an improved linear active disturbance rejection control (ILADRC) strategy is proposed in \cite{11025165}, utilizing an adaptive harmonic filtering ESO to simultaneously mitigate the effects of parameter mismatch and harmonic disturbances on back-EMF estimation}. However, the aforementioned robust observer typically require an increase in the observer's order or the introduction of additional nonlinear terms, which inevitably lead to higher parameter tuning complexity and potential chattering issues \cite{wang2024eso,woldegiorgis2022sensorless,11025165}.

Overall, existing parameter identification and robust observer methods are limited by computational complexity, parameter coupling, structural design, and real-time performance. Thus, there is an urgent need for a controller with a simple structure, efficient computation, decoupled parameters, and superior real-time performance. To address the aforementioned issues, this paper integrates concepts from parameter identification and robust observers to propose a sensorless control method based on a nonlinear flux observer that incorporates equivalent flux compensation (EFC). \textcolor{black}{Specifically, an error dynamics model is established to clarify how resistance, inductance, and flux mismatches influence state estimation. On this basis, the combined effects of parameter deviations are mapped into the flux domain, where the resulting offset components are absorbed into an equivalent flux variable. A concise and computationally efficient update law is proposed to adaptively eliminate these bias terms. Robustness is therefore achieved through flux compensation rather than high-gain disturbance rejection or multi-parameter identification.} In summary, the main contributions of this article are as follows:
\begin{enumerate}

\item An equivalent flux update law is proposed to address parameter mismatch in nonlinear flux observers. The multi-parameter mismatch problem involving resistance, inductance, and flux is transformed into an adaptive flux compensation problem, leading to a structurally simple and decoupled implementation.

\item \textcolor{black}{Without introducing complex nonlinear or high-order robust structures, the proposed scheme achieves complete steady-state compensation for resistance and flux mismatches, and partially mitigates the influence of inductance variations, thereby enhancing robustness under wide parameter and load variations.}

\item A detailed analysis of the effects of parameter mismatch on the error dynamics and stability of flux observers is conducted, providing a theoretical foundation for evaluating system performance under mismatch conditions.

\end{enumerate}

\section{\textcolor{black}{Preliminaries} }

This section briefly outlines the mathematical model of the SPMSM and the structure of the nonlinear observer.

\subsection{Mathematical Model of SPMSM}
 The SPMSM is modeled in $\alpha {-} \beta$ frame as follow:
\begin{equation}
    L_s \dot i_{\alpha \beta} = - R_s i_{\alpha \beta} + v_{\alpha\beta} - e_{\alpha\beta}
\label{eq1}
\end{equation}
where $L_s$ and $R_s$ denote the actual inductance and stator resistance, respectively; $i_{\alpha\beta} = [i_\alpha, i_\beta]^\top$ and $v_{\alpha\beta} = [v_\alpha, v_\beta]^\top$ are the current and voltage measured in the $\alpha$–$\beta$ reference frame. The back-EMF $e_{\alpha\beta}$ has the following form:
\begin{equation}
	e_{\alpha\beta}=\omega\psi_m \begin{bmatrix}-\sin\theta \\\cos\theta \end{bmatrix}
\label{eq2}
\end{equation}
where $\omega$, $\psi_m$, and $\theta$ are the electrical angular velocity, actual permanent magnet flux linkage, and electrical rotor position, respectively. The primary goal of sensorless control is to estimate the rotor speed $\omega$ and position $\theta$.

\subsection{Nonlinear Flux Observer}

This section briefly reviews the fundamental design of the nonlinear flux observer. First, an extended stator flux vector $x$ and its time derivative $\dot{x}$ are defined based on \eqref{eq1} and \eqref{eq2}:
\begin{equation}
    \textcolor{black}{x = L_s i_{\alpha\beta} + \psi_m \begin{bmatrix}\cos\theta \\ \sin\theta\end{bmatrix}, \
    \dot x = -R_s i_{\alpha\beta} + v_{\alpha\beta} =  y }
\label{eq3}
\end{equation}

By defining an auxiliary flux vector $\eta(x) = x - L_s i_{\alpha\beta}$, the geometric constraint is naturally satisfied $\|\eta(x)\|^2 = \psi_m^2$. Based on this inherent property, the nonlinear observer is constructed as follows:
\begin{equation}
	\dot{\hat{x}} = y + \phi \eta(\hat{x}) \left( \psi_{m}^{2} - \|\eta(\hat{x})\|^{2} \right)
\label{eq4}
\end{equation}
where $\hat{x} \in \mathbb{R}^2$ denotes the estimated state, $\eta(\hat{x}) = \hat{x} - L_s i_{\alpha\beta}$, and $\phi > 0$ is the observer gain. 

\textcolor{black}{Finally, based on the geometric relationship in \eqref{eq3}, the estimated position $\hat{\theta}$ is obtained directly as follows:}
\begin{equation}
	\hat{\theta} = \tan^{-1}\left(\frac{\hat{x}_2 - L_s i_\beta}{\hat{x}_1 - L_s i_\alpha}\right)
\label{eq5}
\end{equation}

Furthermore, the error dynamics of the flux observer are considered for the case without parameter mismatches. Defining the observation error as $\tilde{x} = x - \hat{x}$, its derivative follows as $\dot{\tilde{x}} = \dot x - \dot{\hat{x}}$, leading to the following error dynamics:
\begin{equation}
    \dot{\tilde{x}} = - \phi \eta(\hat{x})(\psi_m^2 - \|\eta(\hat{x})\|^2) 
\label{eq8}
\end{equation}

According to \eqref{eq3} and $\eta(\hat{x}) = x-L_si_{\alpha\beta} - \tilde{x}$, we have:
\begin{equation}
    \|\eta(\hat{x})\|^2 = \psi_m ^2 - 2 \psi_m(\tilde{x}_1 \cos \theta + \tilde{x}_2 \sin \theta) + \|\tilde{x}\|^2
\label{eq9}
\end{equation}

Substituting \eqref{eq9} into \eqref{eq8} yields:
\begin{equation}
    \dot{\tilde{x}} = \phi \xi_1(\tilde{x}) \left( \tilde{x} - \psi_m \begin{bmatrix}
        \cos \theta \\
        \sin \theta
    \end{bmatrix} \right)
\label{eq10}
\end{equation}
where $\xi_1(\tilde{x}) = 2 \psi_m (\tilde{x}_1 \cos \theta + \tilde{x}_2 \sin \theta)  - \|\tilde{x}\|^2$.

The error dynamics of the nonlinear flux observer under accurate parameters are described  by  \eqref{eq8}-\eqref{eq10}. Based on \eqref{eq5}, the rotor position can then be estimated. However, many studies and experiments have shown that parameter mismatches, such as errors in resistance, inductance, or flux, can significantly degrade estimation accuracy. Therefore, the central objective of this paper is to address the impact of parameter mismatches on position estimation.

\section{Error Dynamics with Parameter Mismatch}

This section analyzes the error dynamics of the nonlinear flux observer under parameter mismatches and discusses a potential compensation strategy.

\subsection{Nonlinear Flux Observer of Parameter Mismatch}
The parameter deviations are defined as $\Delta R_s = R_s - {R}_z $, $\Delta L_s = L_s - {L}_z $, and $\Delta \psi_m = \psi_m - {\psi}_z $, where the subscript $(\cdot)_z$ denotes the nominal values in the motor, and the others represent the actual motor parameters. The flux observer under parameter mismatches is given as follows:
\begin{equation}
    \left\{ \begin{matrix}
   {{{\dot{x}}}^{*}}=-{{R}_{z}}{{i}_{\alpha \beta }}+{{v}_{\alpha \beta }}={{y}^{*}}  \\[2pt]
   \left\| \eta ({{{\hat{x}}}^{*}}) \right\|^2=\left\| {{{\hat{x}}}^{*}}-{{L}_{z}}{{i}_{\alpha \beta }} \right\|^2={{\psi }_{z}^2}  \\[2pt]
   {{{\dot{\hat{x}}}}^{*}}={{y}^{*}}+\phi ({{{\hat{x}}}^{*}}-{{L}_{z}}{{i}_{\alpha \beta }})\left( \psi _{z}^{2}-{{\left\| \eta ({{{\hat{x}}}^{*}}) \right\|}^{2}} \right)  \\[2pt]
    \textcolor{black}{\hat{\theta}^* = \tan^{-1} \left(({\hat{x}^*_2 - L_z i_\beta})/({\hat{x}^*_1 - L_z i_\alpha}) \right)}
\end{matrix} \right.
\label{eq12}
\end{equation}
\textcolor{black}{where \(\dot{x}^*\), \(y^*\), \(\hat{x}^*\), and \(\eta(\hat{x}^*)\) are the observer states under parameter mismatches, and \(\hat{\theta}^*\) is the estimated position under these mismatches.}

With the definitions $\tilde{x}^* = x - \hat{x}^*$, $\dot{\tilde{x}}^* = \dot{x} - \dot{\hat{x}}^*$, $\tilde{\theta} = \theta - \hat{\theta}$ and $\tilde{\theta}^* = \theta - \hat{\theta}^*$, the following relation is obtained:
\begin{equation}
    \dot{\tilde{x}}^* = -\Delta R_s i_{\alpha \beta } - \phi (\hat {x}^* -  L_z i_{\alpha \beta })(\psi_z^2 - \| \eta(\hat{x}^*) \|^2)
\label{eq14}
\end{equation}

\subsection{Lemma}
\textcolor{black}{To facilitate the subsequent analysis, the following lemmas are presented in advance for clarification.  }

\begin{lemma} \label{lemma 1}
    \textcolor{black}{ For $\forall \ a \in \mathbb{R}$, the following holds:}
    \begin{equation} 
        \textcolor{black}{ |\tanh(a)| \geq\frac{|a|}{1+|a|}  }
\label{eq11}
\end{equation}
\end{lemma}
\begin{proof} 
    \textcolor{black}{The proof of Lemma \ref{lemma 1} is provided in Appendix \ref{proof_Lemma1}.}
\end{proof}

\begin{lemma} \label{lemma 2} 

\textcolor{black}{Under the Park transformation, define $i_d = i_\alpha \cos \theta + i_\beta \sin \theta$ and $\hat{i}_d^* = i_\alpha \cos \hat{\theta}^* + i_\beta \sin \hat{\theta}^*$. The following relationship holds:}
    \begin{equation} \textcolor{black}{\left\{ \begin{matrix}
        Q_1 = \tilde{x} ^\top \left[\cos \theta; \sin \theta \right] = \psi_m (1 - \cos\tilde{\theta}) \\[3pt]
        Q_1^* = \tilde{x}^{* \top} \left[\cos \theta; \sin \theta \right] = \Delta L_s i_d + \psi_m - \psi_z \cos \tilde{\theta}^* \\[3pt]
        Q_2^* = \tilde{x}^{*\top} i_{\alpha\beta} =   \Delta L_s \|i_{\alpha\beta}\|^2 + \psi_m i_d - \psi_z \hat{i}_d^*
\label{A80}
\end{matrix} \right. }
\end{equation}
\end{lemma}
\begin{proof}
    \textcolor{black}{The proof of Lemma \ref{lemma 2} is provided in Appendix \ref{proof_Lemma2}.}
\end{proof}

\subsection{Error Dynamics of Parameter Mismatch}
The expression $\eta(\hat{x}^*)$ is considered, which yields:
\begin{equation}
     \eta(\hat{x}^*)= \hat{x}^* - L_z i_{\alpha \beta }  = \Delta L_s i_{\alpha \beta } + \psi_m \begin{bmatrix} \cos \theta  \\ \sin \theta \end{bmatrix} - \tilde{x}^*
\label{eq15}
\end{equation}
 
{From \eqref{A80} and \eqref{eq15}, we have:
\begin{equation}  \textcolor{black}{
    \begin{aligned}
        \| \eta(\hat{x}^*) \|^2  & = \psi_m^2 + \|\tilde{x}^*\|^2 - 2 \psi_m Q_1^* \\
        & \quad - \Delta L_s^2 \|i_{\alpha\beta}\|^2 + 2 \Delta L_s \psi_z \hat{i}_d^* 
    \end{aligned}
\label{eq16} }
\end{equation}

Substituting \eqref{eq16} into \eqref{eq14} yields:
\begin{equation}
    \begin{aligned}
        \dot{\tilde{x}}^* = &  
        \phi \Big( \tilde{x}^* - \psi_m \begin{bmatrix} \cos \theta \\ \sin \theta  \end{bmatrix}  - \Delta L_s i_{\alpha \beta } \Big) \left( \xi_1^* + \xi_2^* \right) -\Delta R_s i_{\alpha \beta }
    \end{aligned}
\label{eq17}
\end{equation}
where
\begin{equation} \textcolor{black}{
    \xi _1^*  = 2\psi_m Q_1^* - \|\tilde{x}^*\|^2 }
\label{eq18}
\end{equation}
\begin{equation} \textcolor{black}{
    \begin{aligned}
        \xi_2^* = \psi_z^2 - \psi_m^2 + \Delta L_s^2 \|i_{\alpha\beta}\|^2 - 2 \Delta L_s \psi_z \hat{i}_d^*
    \end{aligned}
\label{eq19}}
\end{equation}

Expanding \eqref{eq17} leads to the following expression:
\begin{equation}
    \begin{aligned}
        \dot{\tilde{x}}^* & =  \phi  \xi _1^* \left( \tilde{x}^* - \psi_m \begin{bmatrix} \cos\theta \\ \sin\theta \end{bmatrix} \right) \\
         & \quad + \phi  \xi _2^* \left( \tilde{x}^* - \psi_m \begin{bmatrix} \cos\theta \\ \sin\theta \end{bmatrix} \right) \\
         & \quad - \phi \left(\xi _1^* + \xi _2^*\right) \Delta L_s i_{\alpha \beta } -\Delta R_s i_{\alpha \beta }
    \end{aligned}
\label{eq21}
\end{equation}

Obviously, by comparing \eqref{eq10} and \eqref{eq21}, \textcolor{black}{it can be seen that parameter mismatches introduce additional observation error $\Delta \dot{\tilde{x}}^*$ in the flux observer:}
\begin{equation}
    \begin{aligned}
        \Delta \dot{\tilde{x}}^*
          =  &  \phi  \xi _2^* \left( \tilde{x}^* - \psi_m \begin{bmatrix} \cos\theta \\ \sin\theta \end{bmatrix} \right)  \\
         & - \phi \left(\xi _1^* + \xi _2^*  \right) \Delta L_s i_{\alpha \beta } -\Delta R_s i_{\alpha \beta }
    \end{aligned}
\label{eq22}
\end{equation}

The impact of parameter mismatches, including resistance, inductance, and flux deviations, on the observer state variables is presented in \eqref{eq21} and \eqref{eq22}. Two key characteristics of parameter mismatches in the flux observer can be identified:
\begin{itemize}
    \item The influence of parameter mismatches is not constant but varies with both rotor angle and current due to the underlying parameter deviations, thereby directly compromising the accuracy of angle estimation.
    \item During long-term operation in industrial applications, motor parameters inevitably change over time. These parameter variations increase oscillations and reduce the accuracy in angle estimation. Therefore, addressing the degradation of accuracy caused by parameter mismatches is crucial.
\end{itemize}

\subsection{Compensation Strategy for Parameter Mismatch}
To address the aforementioned issues, this paper proposes transferring the effects of angle estimation degradation caused by parameter inaccuracies to the flux domain and mitigating them by compensating for the equivalent flux.

\textcolor{black}{From \eqref{eq10}, \eqref{eq21} and \eqref{eq22}, the simplest approach is to identify an equivalent error flux $\Delta \psi_z$ to compensate for $\Delta \dot{\tilde{x}}^*$. Evidently, $\Delta \psi_z$ should satisfy the following condition:}
\begin{equation} \textcolor{black}{
     \begin{aligned}
        \dot{\tilde{x}}^* & =  \phi  \xi _1^* \left( \tilde{x}^* - \big(\psi_m + \Delta \psi_z \big) \begin{bmatrix} \cos\theta \\ \sin\theta \end{bmatrix} \right) + \Delta \dot{\tilde{x}}^*  \\
        &  = \phi  \xi _1^* \left( \tilde{x}^* - \psi_m \begin{bmatrix} \cos\theta \\ \sin\theta \end{bmatrix} \right)
\end{aligned}
\label{eq22+1}}
\end{equation}

\textcolor{black}{From \eqref{eq22+1}, we have:}
\begin{equation} \textcolor{black}{
     \phi  \xi _1^* \Delta \psi _z \begin{bmatrix} \cos\theta \\ \sin\theta \end{bmatrix} =  \Delta \dot{\tilde{x}}^* \ \Rightarrow  \ \Delta \psi _z = \eta^\top(x) \frac{ \Delta \dot{\tilde{x}}^*}{ \phi  \xi _1^*  \psi _m}
\label{eq23}}
\end{equation}

\textcolor{black}{Thus, the corresponding equivalent flux should be:}
\begin{equation} \textcolor{black}{
    \hat{\psi }_z = {\psi }_m + \Delta \psi _z = \eta^{\top}(x) \left(  \frac{\phi  \xi _1^* \eta (x) + \Delta \dot{\tilde{x}}^*}{ \phi \xi _1^*\psi _m} \right)
\label{eq25}}
\end{equation}
where $ \hat{\psi }_z $ represents the equivalent flux that combines the nominal flux with the compensation for parameter mismatches.

Unfortunately, the exact value of \eqref{eq25} is unknown. However, by applying the scalar projection formula, the relationship between $\hat{\psi}_z$ and  $\eta(x)$ can be established, allowing \eqref{eq25} to be reformulated as follows:
\begin{equation}
    \hat{\psi }_z =  \|\eta(x)\| \cdot s_p \Rightarrow s_p > 0
\label{eq26}
\end{equation}
\begin{equation}
    s_p = \mathrm{comp}_{\eta(x)} \left( \frac{\phi  \xi _1^* \eta (x) + \Delta \dot{\tilde{x}}^*}{ \phi \xi _1^* \psi _m} \right)
\label{eq27}
\end{equation}
where $s_p = \text{comp}_{(\cdot)}(\star)$ denotes the scalar projection of vector $(\star)$ onto the direction of vector $(\cdot)$.

From \eqref{eq26}, the following approximate relationship can be obtained:
\begin{equation}
    \hat{\psi}_z \propto \|\eta(x)\|  \ \Rightarrow  \ \hat{\psi}_z \propto \|\eta(\hat{x}^*)\|
\label{eq28}
\end{equation}

\textcolor{black}{In practical systems, the estimated flux $\hat{\psi}_z$ remains strictly positive, which ensures that $s_p$ is also positive and bounded. As illustrated in Fig. \ref{Vector_Error}, the residual $\|\eta(x)-\eta(\hat{x}^*)\|$ remains relatively small, compared to fundamental harmonics, even under a $\pm50\%$ step change in resistance at 1000 rpm and full load, confirming the validity of the approximation in \eqref{eq28}.}

 \vspace{-0.3cm}
\begin{figure}[!t] 
    \centering
    \includegraphics[width=0.7\linewidth]{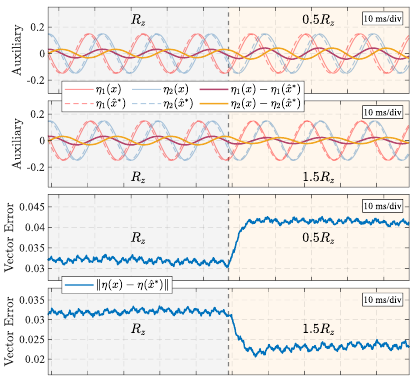} 
    \caption{\textcolor{black}{Dynamic response of $\|\eta(x)-\eta(\hat{x}^*)\|$ in the nonlinear flux observer under resistance variation at 1000 rpm and full load.} }
    \label{Vector_Error}
\end{figure}

\section{Equivalent Flux Update Law}
In the previous section, the relationship governing the equivalent flux was derived. In this section, a flux update law is designed based on a loss function constructed from the geometric structure, considering both the magnitude and directional errors. 

In the $\alpha$-$\beta$ reference frame, the stator flux can be expressed as a two-dimensional vector $\psi_{m,\alpha\beta}$, given by:
\begin{equation}
    \psi_{m,\alpha\beta} = \begin{bmatrix} \psi_{m,\alpha} \\ \psi_{m,\beta} \end{bmatrix}
    = \psi_m \begin{bmatrix} \cos\theta \\ \sin\theta \end{bmatrix}
\label{eq29}
\end{equation}
\textcolor{black}{where $\psi_m=\sqrt{\psi_{m,\alpha}^2+\psi_{m,\beta}^2}$ represents the magnitude, and $n = [\cos\theta; \sin\theta]$ denotes the unit direction vector.
Thus, we have $\psi_{m,\alpha\beta} = \psi_m \cdot n = \eta(x)$.}

First, for the estimated state $\hat{x}^*$, \textcolor{black}{the estimated flux direction $\hat{n}$ and the flux observation can be written as:}
\begin{equation}
    \hat{n} = \frac{\eta(\hat{x}^*)}{\|\eta(\hat{x}^*)\|}
    = \begin{bmatrix} \cos\hat{\theta}^* \\ \sin\hat{\theta}^* \end{bmatrix}, \ \eta(\hat{x}^*) = \hat{\psi}_z \begin{bmatrix} \cos\hat{\theta}^* \\ \sin\hat{\theta}^* \end{bmatrix}
\label{eq30}
\end{equation}

To minimize both magnitude and direction errors, the loss function $\mathcal{T}(\hat{x}^*)$ is designed as:
\begin{equation}
    \begin{aligned}
        \mathcal{T}(\hat{x}^*) & = \frac{1}{2} \big\| \hat{\psi}_z \hat{n} - \eta(\hat{x}^*) \big\|^2 \\
        & = \frac{1}{2} (\| \eta (\hat{x}^*) \| - \hat{\psi }_z  )^2
    \end{aligned}
\label{eq35}
\end{equation}

Next, based on \eqref{eq35}, the following equivalent flux update law is proposed:
\begin{equation}
     \dot{\hat{\psi }}_z= - \hat{s}_{p1} \frac{\partial \mathcal{T} (\hat{x}^*)}{\partial \hat{\psi }_z} = \hat{s}_{p1} (\| \eta (\hat{x}^*)  \| - \hat{\psi }_z )
\label{eq36}
\end{equation}
where $\hat{s}_{p1} >0$ denotes the gain of the equivalent flux update law. This formulation achieves linear asymptotic convergence by simultaneously regulating both the magnitude and direction of the flux.

To further enhance the stability and convergence of the flux update law, \eqref{eq36} is modified into a nonlinear update form with the following saturation characteristics:

\begin{equation}
    \dot{\hat{\psi }}_z = \hat{s}_{p1} \tanh \left(\hat{s}_{p2}(\| \eta (\hat{x}^*) \| - \hat{\psi }_z) \right)  
\label{eq37}
\end{equation}
where $\hat{s}_{p2} >0 $ is also a gain for the equivalent flux update law. In the subsequent stability analysis, it is proven that \(\hat{s}_{p1}\) and \(\hat{s}_{p2}\) accelerate the convergence of the equivalent flux.

The improved nonlinear equivalent flux update law offers the following advantages:
\begin{itemize}
    \item First, the improved update law \eqref{eq37} ensures gradient continuity, improving numerical stability and avoiding saturation dead zones. This update mechanism improves the convergence and robustness of equivalent flux estimation.
    \item Second, during startup or transient operating conditions, significant estimation errors and disturbances may occur. The linear update law \eqref{eq36} can lead to divergence or overshoot of $\dot{\psi}_z$, resulting in flux estimation errors. The $\tanh(\cdot)$ function provides near-linear accuracy for small errors while limiting the update rate for large errors, effectively preventing oscillations and overshoot.
\end{itemize}

The overall block diagram and the algorithmic flowchart of the proposed control strategy are presented in Fig.~\ref{Structure diagram} and Algorithm \ref{alg:flux_observer}, respectively.

\begin{figure}[!t]
    \centering
    \includegraphics[width=0.97\linewidth]{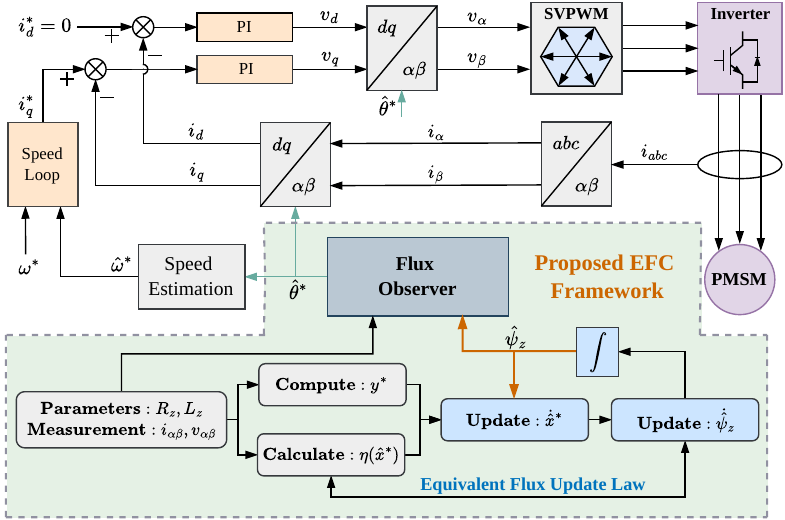} 
    \caption{\textcolor{black}{Proposed control scheme based on equivalent flux update law.} }
    \label{Structure diagram}
\end{figure}

\begin{algorithm}[!t]
\caption{Proposed EFC Method}
\label{alg:flux_observer}
\begin{algorithmic}[1]
\State \textbf{Motor Parameters:} $L_z$, $R_z$
\State \textbf{Control Input:} $\phi$, $\hat{s}_{p1}$, $\hat{s}_{p2}$
\State \textbf{Output:} Position $\hat{\theta}^*$
\Repeat
    \State Measure currents and voltages: $i_{\alpha\beta}$, $v_{\alpha\beta}$
    \State Compute variable: \\
    \qquad \qquad $y^* = -R_z i_{\alpha\beta} + v_{\alpha\beta}$
    \State Calculate state: \\
    \qquad \qquad $\eta(\hat{x}^*) = \hat{x}^*-L_z i_{\alpha\beta}$

    \State Update flux observer: \\
    \qquad \qquad ${{{\dot{\hat{x}}}}^{*}}={{y}^{*}}+\phi ({{{\hat{x}}}^{*}}-{{L}_{z}}{{i}_{\alpha \beta }})\left( \hat{\psi} _{z}^{2}-{{\left\| \eta ({{{\hat{x}}}^{*}}) \right\|}^{2}} \right)$
    \State Equivalent flux update law: \\ \vspace{0.2em} 
    \qquad \qquad $\dot{\hat{\psi }}_z = \hat{s}_{p1} \tanh \left(\hat{s}_{p2}(\| \eta (\hat{x}^*) \| - \hat{\psi }_z) \right)$ 
    \State Calculate $ \hat{x}^*$ and $\hat{\psi}_z$  \vspace{0.1em}
    \State Estimate $\hat{\theta}^*$  \\ \vspace{0.1em}
    \qquad \qquad $\hat{\theta}^* = \tanh^{-1} \left(({\hat{x}^*_2 - L_z i_\beta})/({\hat{x}^*_1 - L_z i_\alpha}) \right)$
\Until{Motor stop}
\end{algorithmic}
\end{algorithm}

\section{Stability Analysis}
This section primarily employs the direct Lyapunov method to analyze the stability of the flux observer under parameter mismatches, as well as the stability and convergence of the equivalent flux update law.

\subsection{Stability Analysis Under Nominal Parameter Conditions}
Before conducting the stability analysis under parameter mismatches, it is necessary to first establish the system's stability under the nominal parameter conditions.

\textcolor{black}{A Lyapunov function is constructed as:}
    \begin{equation} \textcolor{black}{
        V_1 = \frac{1}{2} \|\tilde{x}\|^2 = \frac{1}{2} (\tilde{x}_1^2 + \tilde{x}_2^2)
    \label{eq39}}
    \end{equation}

\begin{theorem} \label{Theorem 1} \textcolor{black}{
    For the SPMSM system described by \eqref{eq1}, \eqref{eq2} with the control strategy \eqref{eq4} and the update law \eqref{eq3}-\eqref{eq5} under the nominal parameters, the variable $\tilde{x}$ is uniformly bounded. The $\tilde{x}$ converges to a compact set $\Omega_1 $, defined as follows:}
\begin{equation}
    \textcolor{black}{\Omega_1 = \left\{ \tilde{x} \in \mathbb{R}^2  \mid \|\tilde{x}\| \leq 2 \psi_m \right\}
\label{eq38}}
\end{equation} 
\end{theorem}

\begin{proof} \textcolor{black}{
    The proof of Theorem \ref{Theorem 1} is provided in Appendix \ref{proof_1}. }
\end{proof}

\subsection{Stability Analysis Under Parameter Mismatch}

This subsection analyzes the stability of the flux observer under parameter mismatches and discusses the effects of resistance, inductance, and flux mismatches on system stability. \textcolor{black}{A new Lyapunov function is constructed as:}
    \begin{equation}\textcolor{black}{
        V_2 = \frac{1}{2} \|\tilde{x}^*\|^2 = \frac{1}{2}(\tilde{x}_1^{*2} + \tilde{x}_2^{*2})
    \label{eq45}}
    \end{equation}

\begin{theorem} \label{Theorem 2} \textcolor{black}{
    For the SPMSM system described by \eqref{eq1} and \eqref{eq2} with the control strategy \eqref{eq12} under parameter mismatches, the stability of the flux observer is governed by parameter accuracy. Excessive deviations in flux, inductance, or resistance will violate the Lyapunov stability criteria, leading to increased steady-state errors or observer divergence. }
\end{theorem}

\begin{proof}\textcolor{black}{
    The proof of {Theorem \ref{Theorem 2}} is provided in Appendix \ref{proof_2}.}
\end{proof}

\textcolor{black}{It should be noted from the proof of Theorem \ref{Theorem 2} that parameter mismatch does not necessarily exert a detrimental effect on observer performance. Under certain operating conditions, an appropriate increase or decrease in the nominal values of resistance and inductance can reduce the convergence error $\|\tilde{x}^*\|$. This compensatory effect is intrinsically linked to the system’s operating point, specifically the load conditions, the stator current $i_{\alpha\beta}$, and the sign and magnitude of the $d$-axis current deviation $(i_d - \hat{i}_d^*)$.}

\subsection{Stability Analysis of the Equivalent Flux Update Law}
This part analyzes the stability of the equivalent flux update law. First, define the flux error variable $\tilde{\psi}_z = \hat{\psi}_z - \|\eta(\hat{x}^*)\|$. \textcolor{black}{A new Lyapunov function is constructed as:}
    \begin{equation}\textcolor{black}{
        V_3 = \frac{1}{2} \tilde{\psi }_z^2
    \label{eq59}}
    \end{equation}

\begin{theorem} \label{Theorem 3}\textcolor{black}{
    For the SPMSM system described by \(\eqref{eq1}\), \(\eqref{eq2}\) with the control strategy \(\eqref{eq12}\) and equivalent flux update law \eqref{eq37} under parameter mismatches, the auxiliary variable $\tilde{\psi}_z$ is uniformly bounded. The $\tilde{\psi}_z$ converges to a compact set $\Omega_2 $, defined as follows:}
\begin{equation}\textcolor{black}{
    \Omega_2 = \left\{ \tilde{\psi}_z \in \mathbb{R}  \mid |\tilde{\psi}_z| \leq \frac{\Psi}{\hat{s}_{p2} (\hat{s}_{p1}- \Psi)} \right\}
\label{eq58} }
\end{equation}
\textcolor{black}{ where $\Psi = \sup \|\dot \eta(\hat{x}^*)\|$ is bounded.}
\end{theorem}
\begin{proof}
   \textcolor{black}{ The proof of {Theorem \ref{Theorem 3}} is provided in Appendix \ref{proof_3}.}
\end{proof}

Since \(|\tilde{\psi}_z| > 0\), from \eqref{eq58}, \(\hat{s}_{p1}\) should be chosen with a high gain to satisfy \(\hat{s}_{p1} - \Psi > 0\). \textcolor{black}{In practice, the exact value of $\Psi$ is unnecessary, as an appropriate gain $\hat{s}_{p1}$ can be easily obtained through simple parameter tuning.}

\begin{figure}[!t]
    \centering
        \includegraphics[width=0.65\linewidth]{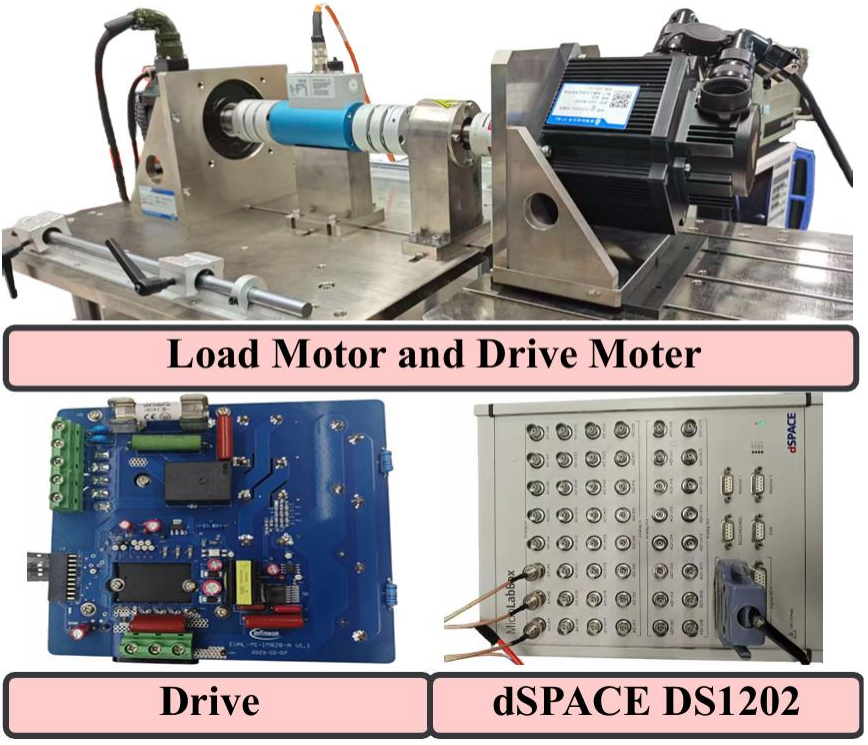} 
    \caption{ Experimental platform of SPMSM.} 
    \label{Motor Platform}
\end{figure}

\begin{table}[!t]
  \caption{Parameters of the SPMSM System}
  \renewcommand\arraystretch{1.25} 
  \small
  \centering
  \setlength{\tabcolsep}{5mm} 
  \begin{tabular}{c c c } \hline \hline
    Description & Parameter & Value   \\ \hline 
    Rated power & $P_r$ & 1 kW \\
    External torque & $T_e$ & 4 N $\!\cdot\!$ m \\
    Number of pole pairs & $N_p$ & 4 poles  \\
    DC voltage & $U_{dc}$ & 300 V \\ \hline
    Nominal inductance & $L_z$ & 3.21 mH  \\
    Nominal stator resistance & $R_z$ & 1.38 $\Omega$ \\
    Nominal flux linkage & $\psi_z $ & 0.148 Wb   \\ \hline
    Observer gain & $\phi$ & 8000  \\
    Update law gains & $(\hat{s}_{p1}, \hat{s}_{p2})$ & (10,10) \\
    \hline  \hline 
  \end{tabular}
  \label{Table.Parameters}
\end{table}

\section{Experimental Results and Analysis}

To validate the proposed equivalent flux update law, a sensorless SPMSM experimental platform is built, as shown in Fig. \ref{Motor Platform}. It consists of two PMSMs, a dSPACE DS1202 system, a driver module, and a host computer. The speed loop runs at 2 kHz, and both the current loop and PWM switching frequency are 10 kHz. Key SPMSM parameters are listed in Table \ref{Table.Parameters}.

\begin{figure}[!t]
    \centering
    \setlength{\subfigbottomskip}{-2pt} 
    \setlength{\subfiglabelskip}{-10pt} 
    \setlength{\subfigcapskip}{-6pt}   
    \subfigure[]{
        \includegraphics[width=0.48\linewidth]{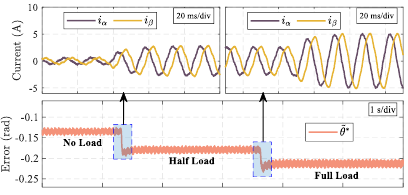}
    }
     \hspace{-0.05\linewidth} 
     \subfigure[]{
        \includegraphics[width=0.48\linewidth]{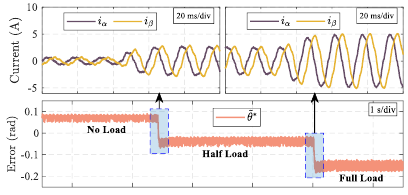}
    }
    \subfigure[]{
        \includegraphics[width=0.48\linewidth]{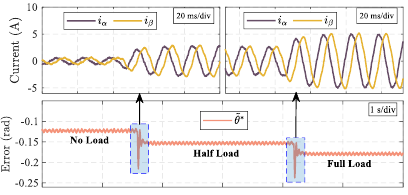}
    }
     \hspace{-0.05\linewidth} 
    \subfigure[]{
        \includegraphics[width=0.48\linewidth]{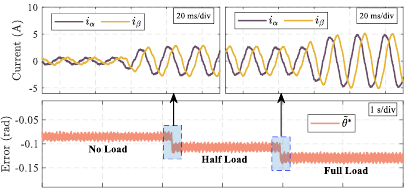}
    }
    \caption{\textcolor{black}{Experimental results under no-load, half-load, and full-load conditions at 1000 rpm. (a) Method 1. (b) Method 2. (c) Method 3. (d) Proposed EFC.}}
    \label{half_load}
\end{figure}

\subsection{Experimental Setup}

\textcolor{black}{To evaluate the control performance under parameter uncertainties, comparative experiments are conducted by introducing $\pm$50\% step variations and sinusoidal mismatches to the resistance, inductance, and flux. Three representative controllers are compared: Method 1, the nonlinear flux observer \cite{lee2009sensorless}; Method 2, an EKF-based controller \cite{10949742}; and Method 3, a linear flux observer-based controller \cite{zhou2023robust_1}.}

\textcolor{black}{The control framework relies on three independent tuning parameters, namely one observer gain $\phi$ and two update law gains $(\hat{s}_{p1},\hat{s}_{p2})$. These gains are tuned by first optimizing the observer gain and then the update law gains. Each parameter is determined through a systematic iterative process, progressing from coarse to fine tuning, and the finalized values are summarized in Table \ref{Table.Parameters}.} For position error visualization, the $\pm 2\pi$ phase-wrapping outliers are filtered and replaced with the preceding data to ensure clarity without affecting the analysis.

\begin{table}[!t] \color{black}
  \caption{RMSE of Estimation Errors under Different Parameter Mismatch Conditions. Unit: rad}
  \renewcommand\arraystretch{1.25} 
  \small
  \centering
  \setlength{\tabcolsep}{2mm} 
  \begin{tabular}{c c c c c c} \hline \hline
     & Method 1 & Method 2 & Method 3 & Proposed EFC   \\ \hline 
    $R_z$ & $0.215$ & $0.174$ & $0.17$9 & $ \underline{0.129}$ \\
    $0.5R_z$ & $0.279$ & $0.207$ &$ 0.197$ & $\underline{0.128}$ \\
    $1.5R_z$ & $0.157$ & $0.145$ & $0.163$ & $\underline{0.130}$ \\ \hline
    $L_z$ & $0.215$ & $0.174$ &$ 0.179$ & $\underline{0.129}$ \\
    $0.5L_z$  & $0.261$ & $\underline{0.131}$ & $0.233$ & $0.181$  \\
    $1.5L_z$  & $0.172$ & $0.219$ & $0.128$ & $\underline{0.078}$   \\ \hline
    $\psi_z$  & $0.215$ & $0.174$ & $0.179$ & $\underline{0.129}$   \\
    $0.5\psi_z$ & $0.532$ & $0.684$ & $0.332$ & $\underline{0.129}$  \\
    $1.5\psi_z$ & $0.222$ & $0.096$ & $\underline{0.023}$ & $0.130$  \\
    \hline  \hline 
  \end{tabular}
  \label{Table.RMSE}
\end{table}

\begin{figure}[!t]
    \centering
    \setlength{\subfigbottomskip}{-2pt} %
    \setlength{\subfiglabelskip}{-10pt} %
    \setlength{\subfigcapskip}{-6pt}   %
    \subfigure[]{
        \includegraphics[width=0.48\linewidth]{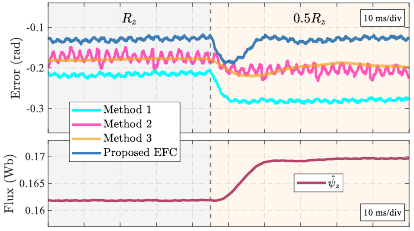}
    }
     \hspace{-0.05\linewidth} 
    \subfigure[]{
        \includegraphics[width=0.48\linewidth]{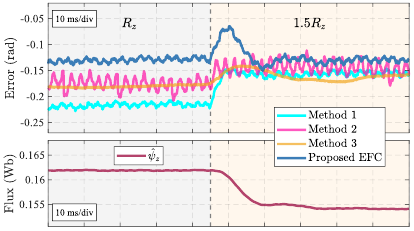}
    }
    \caption{\textcolor{black}{Experimental results under resistance mismatch at 1000 rpm and full load, showing the flux $\hat{\psi}_z$ update in the proposed EFC. (a) $R_z \rightarrow 0.5R_z$. (b) $R_z \rightarrow 1.5R_z$.}}
    \label{R_0.5R_1.5R}
\end{figure}

\begin{figure}[!t]
    \centering
    \setlength{\subfigbottomskip}{-2pt} %
    \setlength{\subfiglabelskip}{-10pt} %
    \setlength{\subfigcapskip}{-6pt}   %
    \subfigure[]{
        \includegraphics[width=0.48\linewidth]{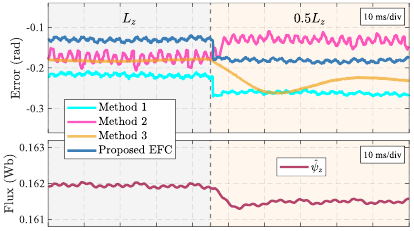}
    }
     \hspace{-0.05\linewidth} 
    \subfigure[]{
        \includegraphics[width=0.48\linewidth]{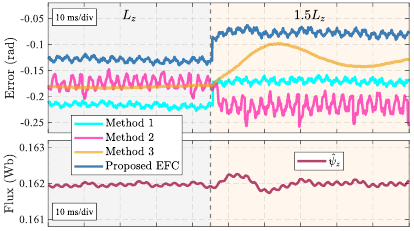}
    }
    \caption{\textcolor{black}{Experimental results under inductance mismatch at 1000 rpm and full load, showing the flux $\hat{\psi}_z$ update in the proposed EFC. (a) $L_z \rightarrow 0.5L_z$. (b) $L_z \rightarrow 1.5L_z$.}}
    \label{L_0.5L_1.5L}
\end{figure}

\begin{figure}[!t]
    \centering
    \setlength{\subfigbottomskip}{-2pt} %
    \setlength{\subfiglabelskip}{-10pt} %
    \setlength{\subfigcapskip}{-6pt}   %
    \subfigure[]{
        \includegraphics[width=0.48\linewidth]{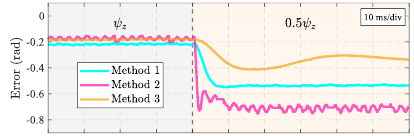}
    }
     \hspace{-0.05\linewidth} 
    \subfigure[]{
        \includegraphics[width=0.48\linewidth]{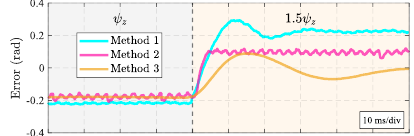}
    }
    \subfigure[]{
        \includegraphics[width=0.48\linewidth]{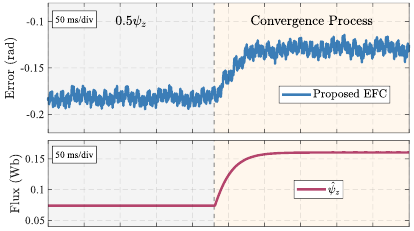}
    }
    \hspace{-0.05\linewidth} 
    \subfigure[]{
        \includegraphics[width=0.48\linewidth]{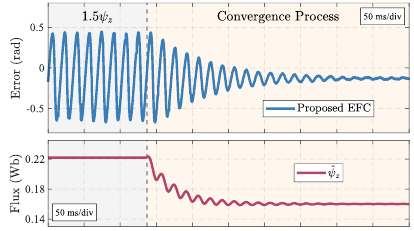}
    }
    \caption{\textcolor{black}{Experimental results under flux mismatch at 1000 rpm and full load. (a) $\psi_z \rightarrow 0.5\psi_z$. (b) $\psi_z \rightarrow 1.5\psi_z$. (c) Update of $\hat{\psi}_z$ in proposed EFC with initial $0.5\psi_z$. (d) Update of $\hat{\psi}_z$ in proposed EFC with initial $1.5\psi_z$.}}
    \label{psi_0.5psi_1.5psi}
\end{figure}

\subsection{Performance Under Variable Load}

\textcolor{black}{To evaluate the steady-state accuracy and robustness of the algorithm under different loading conditions, comparative experiments are conducted at 1000 rpm with transitions among no-load, half-load, and full-load operation. The experimental results are shown in Fig. \ref{half_load}.}

\textcolor{black}{Under no-load and half-load conditions, Method 2 and the proposed EFC demonstrate superior accuracy, bounding the position error within $\pm0.11$ rad. As the load increases to full-load, the proposed EFC significantly outperforms the other methods, achieving the lowest estimation error of approximately $-0.13$ rad, compared to $-0.22$ rad, $-0.17$ rad, and $-0.18$ rad for Methods 1, 2, and 3, respectively. Although estimation errors rise with loading for all methods, the proposed EFC consistently maintains a relatively low steady-state error throughout the loading transitions, validating its robustness against load torque disturbances.}

\subsection{Performance Under Step Parameter Mismatch}

\textcolor{black}{To evaluate the robustness of the proposed method against parameter mismatch, comparative experiments are performed at 1000 rpm under full-load conditions, focusing on step variations in resistance, inductance, and flux.}

\textcolor{black}{Fig. \ref{R_0.5R_1.5R} shows the results under resistance mismatch. At $0.5R_z$, Method 1 exhibits the largest error of approximately $-0.28$ rad, while at $1.5R_z$, Method 2 achieves a smaller error of about $-0.15$ rad. In contrast, the proposed EFC maintains the smallest error near $-0.13$ rad, and recovers within 20 ms in both cases, showing effective resistance mismatch compensation.}

\textcolor{black}{The influence of inductance variation is presented in Fig. \ref{L_0.5L_1.5L}. For $0.5L_z$, Method 1 produces the largest error, reaching approximately $-0.26$ rad, while Method 2 shows a smaller error. For $1.5L_z$, the proposed EFC achieves the smallest error, about $-0.07$ rad. The proposed method reduces sensitivity to inductance variation, although the improvement is less pronounced than in the resistance case.}

\textcolor{black}{Fig. \ref{psi_0.5psi_1.5psi} illustrates the results under flux variation, where the influence is more pronounced than that of resistance and inductance. At $0.5\psi_z$, Method 2 produces the largest error of about $-0.71$ rad, while Method 3 shows a smaller deviation at $1.5\psi_z$. In contrast, the proposed EFC converges to the same steady-state error of $-0.13$ rad and flux of 0.16 Wb under both conditions, although the convergence at $1.5\psi_z$ is slower.}

\textcolor{black}{Overall, the proposed method exhibits strong robustness to resistance and flux mismatch, while also providing partial mitigation of inductance sensitivity. The root mean square errors (RMSE) calculated over a 2s steady-state period for parameter mismatch are summarized in Table \ref{Table.RMSE}.}

\begin{figure}[!t]
    \centering
    \setlength{\subfigbottomskip}{-2pt} %
    \setlength{\subfiglabelskip}{-10pt} %
    \setlength{\subfigcapskip}{-6pt}   %
    \subfigure[]{
        \includegraphics[width=0.95\linewidth]{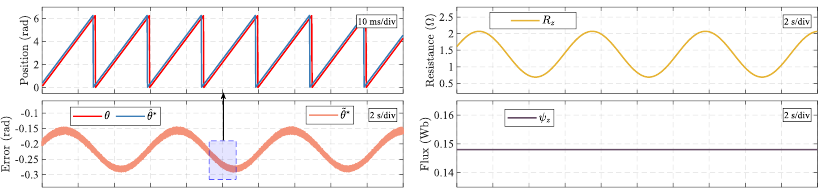}
    }
     \hspace{-0.1\linewidth} 
     \subfigure[]{
        \includegraphics[width=0.95\linewidth]{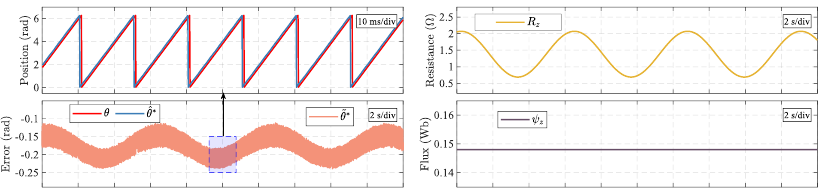}
    }
     \hspace{-0.1\linewidth} 
    \subfigure[]{
        \includegraphics[width=0.95\linewidth]{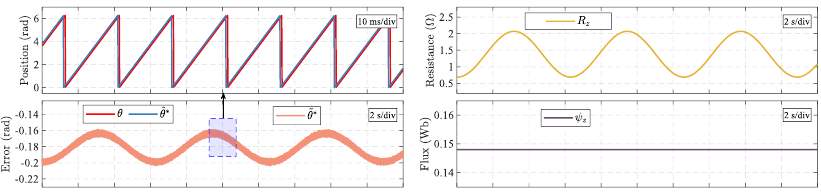}
    }
     \hspace{-0.1\linewidth} 
    \subfigure[]{
        \includegraphics[width=0.95\linewidth]{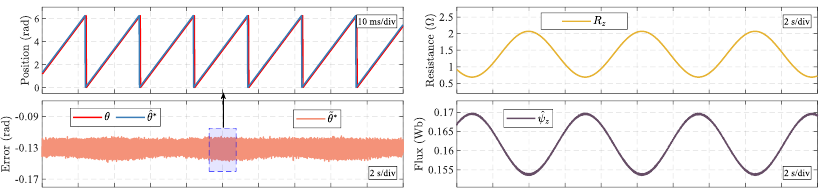}
    }
    \caption{\textcolor{black}{Experimental results under resistance mismatch $R_z + 0.5R_z\sin(2\pi t)$ at 1000 rpm and full load. (a) Method 1. (b) Method 2. (c) Method 3. (d) Proposed EFC.}}
    \label{sin_0.5R_1.5R}
\end{figure}

\begin{figure}[!t]
    \centering
    \setlength{\subfigbottomskip}{-2pt} %
    \setlength{\subfiglabelskip}{-10pt} %
    \setlength{\subfigcapskip}{-6pt}   %
    \subfigure[]{
        \includegraphics[width=0.95\linewidth]{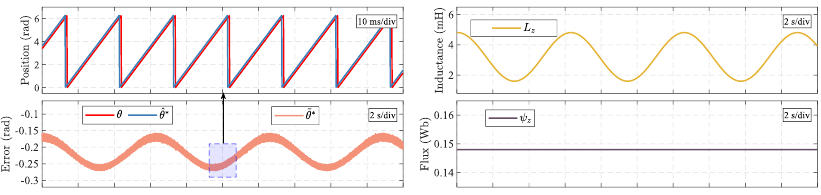}
    }
     \hspace{-0.1\linewidth} 
     \subfigure[]{
        \includegraphics[width=0.95\linewidth]{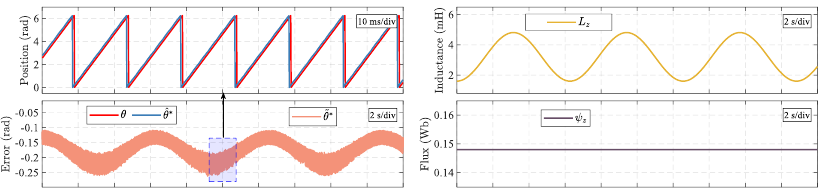}
    }
     \hspace{-0.1\linewidth}
    \subfigure[]{
        \includegraphics[width=0.95\linewidth]{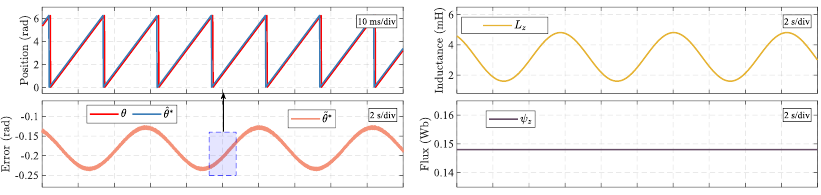}
    }
     \hspace{-0.1\linewidth} 
    \subfigure[]{
        \includegraphics[width=0.95\linewidth]{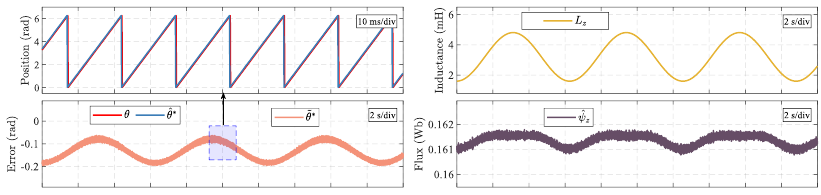}
    }
    \caption{\textcolor{black}{Experimental results under inductance mismatch $L_z + 0.5L_z\sin(2\pi t)$ at 1000 rpm and full load. (a) Method 1. (b) Method 2. (c) Method 3. (d) Proposed EFC.}}
    \label{sin_0.5L_1.5L}
\end{figure}

\begin{figure}[!t]
    \centering
    \setlength{\subfigbottomskip}{-2pt} %
    \setlength{\subfiglabelskip}{-10pt} %
    \setlength{\subfigcapskip}{-6pt}   %
    \subfigure[]{
        \includegraphics[width=0.95\linewidth]{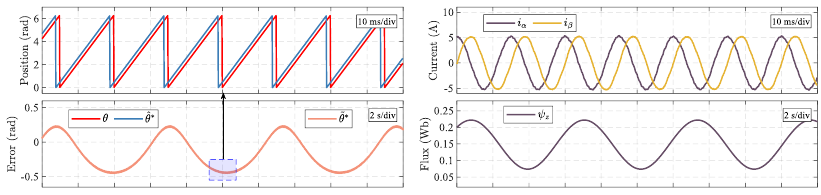}
    }
     \hspace{-0.1\linewidth} 
     \subfigure[]{
        \includegraphics[width=0.95\linewidth]{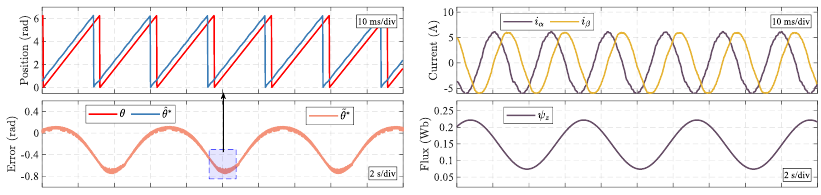}
    }
    \hspace{-0.1\linewidth} 
    \subfigure[]{
        \includegraphics[width=0.95\linewidth]{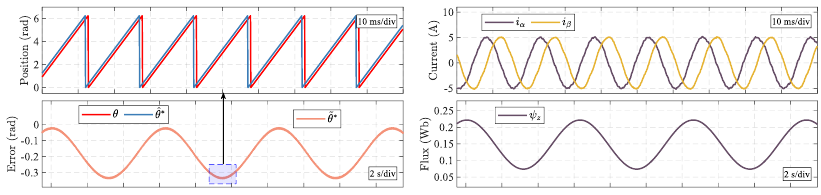}
    }
    \hspace{-0.1\linewidth} 
    \subfigure[]{
        \includegraphics[width=0.96\linewidth]{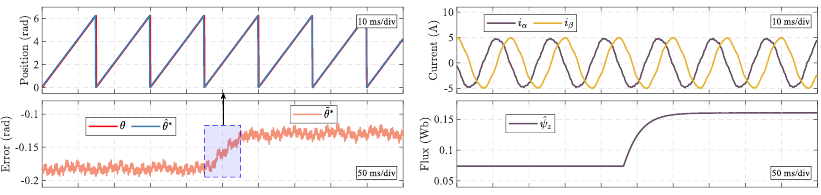}
    }
    \hspace{-0.1\linewidth} 
    \subfigure[]{
        \includegraphics[width=0.96\linewidth]{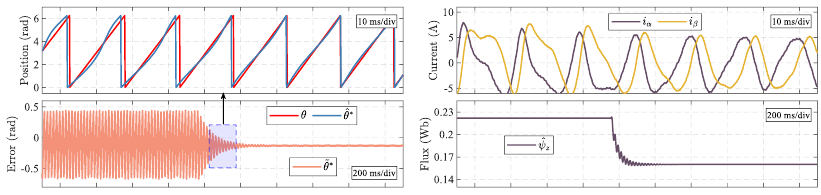}
    }
    \caption{\textcolor{black}{Experimental results under flux mismatch $\psi_z + 0.5\psi_z\sin(2\pi t)$ at 1000 rpm and full load. (a) Method 1. (b) Method 2. (c) Method 3. (d) Update of $\hat{\psi}_z$ in proposed EFC with initial $0.5\psi_z$. (e) Update of $\hat{\psi}_z$ in proposed EFC with initial $1.5\psi_z$.}}
    \label{sin_0.5psi_1.5psi}
\end{figure}

\begin{figure}[!t]
    \centering
    \setlength{\subfigbottomskip}{-2pt} %
    \setlength{\subfiglabelskip}{-10pt} %
    \setlength{\subfigcapskip}{-6pt}   %
    \subfigure[]{
        \includegraphics[width=0.48\linewidth]{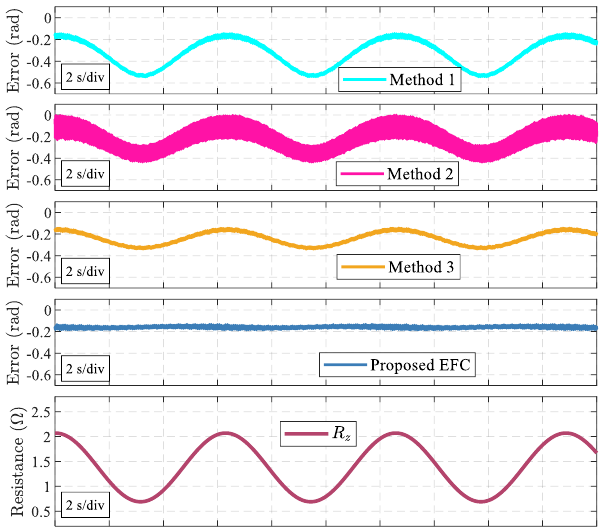}
    }
    \hspace{-0.05\linewidth}
    \subfigure[]{
        \includegraphics[width=0.48\linewidth]{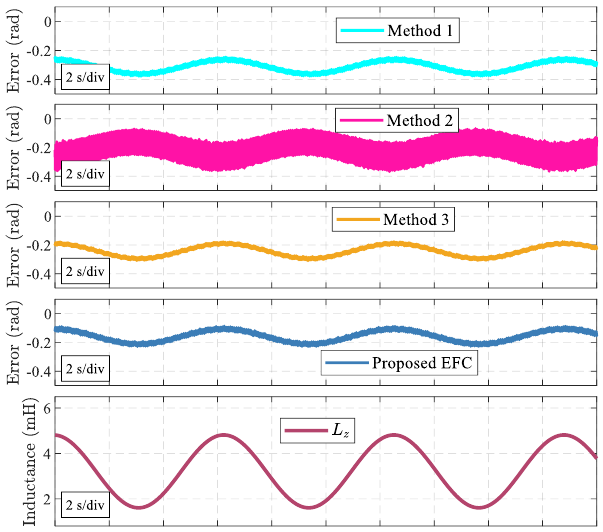}
    }
    \hspace{-0.05\linewidth}
    \subfigure[]{
        \includegraphics[width=0.48\linewidth]{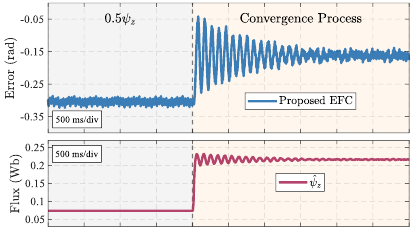}
    }
    \hspace{-0.05\linewidth}
    \subfigure[]{
        \includegraphics[width=0.48\linewidth]{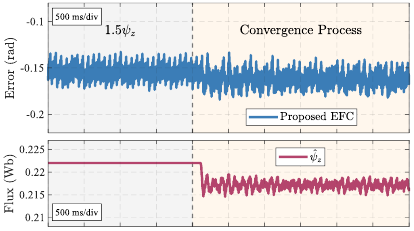}
    }
    \caption{\textcolor{black}{Experimental results under parameter mismatch at 200 rpm and full load. (a) Resistance mismatch $R_z + 0.5R_z\sin(2\pi t)$. (b) Inductance mismatch $L_z + 0.5L_z\sin(2\pi t)$. (c) Update of $\hat{\psi}_z$ in proposed EFC with initial $0.5\psi_z$. (d) Update of $\hat{\psi}_z$ in proposed EFC with initial $1.5\psi_z$.}}
    \label{200_0.5R_1.5R}
\end{figure}

\begin{figure}[!t]
    \centering
    \setlength{\subfigbottomskip}{-2pt} %
    \setlength{\subfiglabelskip}{-10pt} %
    \setlength{\subfigcapskip}{-6pt}   %
    \subfigure[]{
        \includegraphics[width=0.48\linewidth]{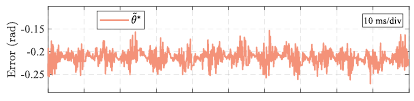}
    }
     \hspace{-0.05\linewidth} 
     \subfigure[]{
        \includegraphics[width=0.48\linewidth]{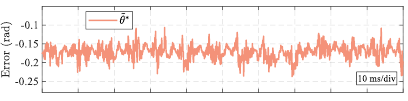}
    }
    \subfigure[]{
        \includegraphics[width=0.48\linewidth]{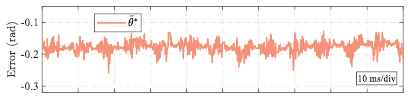}
    }
     \hspace{-0.05\linewidth} 
    \subfigure[]{
        \includegraphics[width=0.48\linewidth]{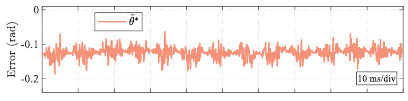}
    }
    \caption{\textcolor{black}{Experimental results under Gaussian noise $N(0, 0.5)$ at 1000 rpm and full load. (a) Method 1. (b) Method 2. (c) Method 3. (d) Proposed EFC.}}
    \label{noise}
\end{figure}

\begin{figure}[!t]
    \centering
    \setlength{\subfigbottomskip}{-2pt} %
    \setlength{\subfiglabelskip}{-10pt} %
    \setlength{\subfigcapskip}{-6pt}   %
    \subfigure[]{
        \includegraphics[width=0.48\linewidth]{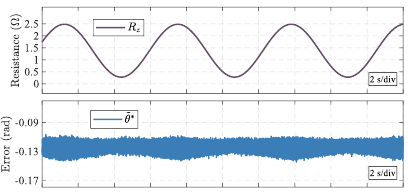}
    }
     \hspace{-0.05\linewidth} 
    \subfigure[]{
        \includegraphics[width=0.48\linewidth]{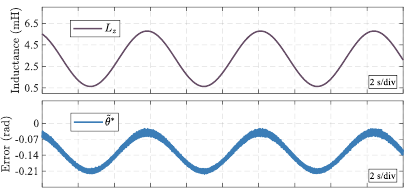}
    }
    \subfigure[]{
        \includegraphics[width=0.48\linewidth]{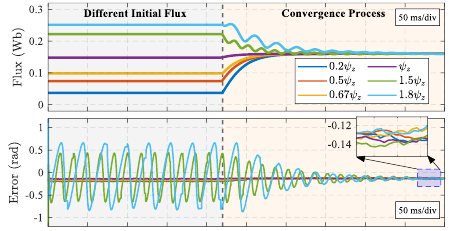}
    }
     \hspace{-0.05\linewidth} 
    \subfigure[]{
        \includegraphics[width=0.48\linewidth]{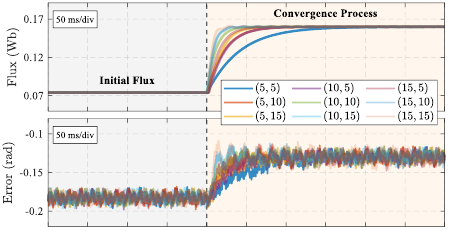}
    }
    \caption{\textcolor{black}{Experimental results of the proposed EFC under parameter mismatch and initial value sensitivity at 1000 rpm and full load. (a) Resistance mismatch $R_z + 0.8R_z\sin(2\pi t)$. (b) Inductance mismatch $L_z + 0.8L_z\sin(2\pi t)$. (c) Update of $\hat{\psi}_z$ with various initial flux. (d) Update of $\hat{\psi}_z$ under different gains.}}
    \label{0.2R_1.2R}
\end{figure}

\begin{figure}[!t]
    \centering
    \setlength{\subfigbottomskip}{-2pt} %
    \setlength{\subfiglabelskip}{-10pt} %
    \setlength{\subfigcapskip}{-6pt}   %
    \subfigure[]{
        \includegraphics[width=0.48\linewidth]{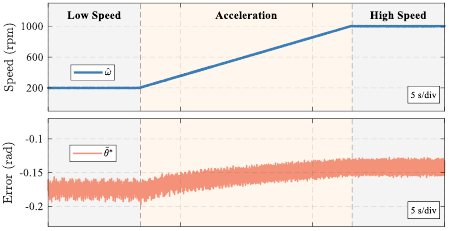}
    }
     \hspace{-0.05\linewidth} 
    \subfigure[]{
        \includegraphics[width=0.48\linewidth]{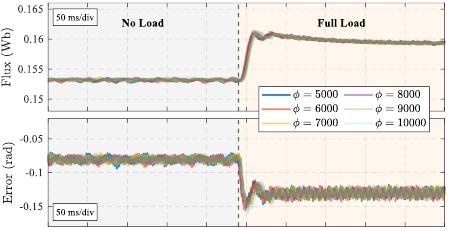}
    }
    \caption{\textcolor{black}{Experimental results of the proposed EFC under full load. (a) Acceleration process from 200 rpm to 1000 rpm. (b) Update of $\hat{\psi}_z$ at 1000 rpm with different observer gains. }}
    \label{gamma}
\end{figure}

\subsection{Performance Under Sinusoidal Parameter Mismatch}

\textcolor{black}{To further evaluate the robustness against time-varying parameter mismatches, sinusoidal variations in resistance, inductance, and flux are introduced at 1000 rpm under full-load conditions. The extreme estimation errors are consistent with those observed in the step mismatch tests.}

\textcolor{black}{Fig.~\ref{sin_0.5R_1.5R} shows the results under sinusoidal resistance variation $R_z + 0.5R_z\sin(2\pi t)$. Methods 1-3 exhibit periodic errors synchronized with the resistance change, whereas the proposed EFC maintains a nearly constant angle error around $-0.13$ rad with only small fluctuations. Fig.~\ref{sin_0.5L_1.5L} shows the results under inductance variation $L_z + 0.5L_z\sin(2\pi t)$. Methods 1, 3, and the proposed EFC show periodic errors correlated with the inductance variation, while Method 2 exhibits an inverse trend.}

\textcolor{black}{Fig.~\ref{sin_0.5psi_1.5psi} illustrates the results under sinusoidal flux variation $\psi_z + 0.5\psi_z\sin(2\pi t)$. Methods 1-3 again produce periodic errors with the flux variation. In contrast, the proposed EFC converges to the same steady-state angle error of approximately $-0.13$ rad under both $0.5\psi_z$ and $1.5\psi_z$ initial conditions, although larger transient fluctuations occur for $1.5\psi_z$.}

\subsection{Performance Under Low-Speed and Noise Immunity}

\textcolor{black}{To validate the proposed method under low-speed conditions, additional experiments are conducted at 200 rpm and full load with sinusoidal variations in resistance, inductance, and flux. As shown in Fig. \ref{200_0.5R_1.5R}, the error trends under $R_z$, $L_z$, and $\psi_z$ mismatches are consistent with those at 1000 rpm, while all methods exhibit larger estimation errors due to the reduced back-EMF at low speed.}

\textcolor{black}{Furthermore, Gaussian noise $N(0,0.5)$ is injected into the $\alpha$-$\beta$ currents at 1000 rpm and full load to evaluate noise immunity. As shown in Fig. \ref{noise}, all methods keep bounded estimation errors. High-frequency fluctuations appear in the estimated position, but the steady-state bias remains stable, confirming acceptable noise immunity for all methods.}

\textcolor{black}{Overall, Figs. \ref{R_0.5R_1.5R}-\ref{noise} show that the proposed EFC method compensates for parameter mismatches through equivalent flux. The mismatch-induced flux distortion is effectively absorbed into the equivalent flux term.}

\subsection{Analysis of Initial Sensitivity and Convergence Dynamics }
\textcolor{black}{To further verify the stability and robustness of the proposed EFC, experiments are conducted under larger parameter mismatches, different initial flux, and varying observer gains.}

\textcolor{black}{Fig. \ref{0.2R_1.2R} presents the results at 1000~rpm and full load with larger sinusoidal variations $R_z + 0.8R_z \sin(2\pi t)$ and $L_z + 0.8L_z \sin(2\pi t)$. Under resistance mismatch, the estimation error remains around $-0.13$ rad. Under inductance mismatch, the error increases slightly. For initial flux values ranging from $0.2\psi_z$ to $1.8\psi_z$, $\hat{\psi}z$ converges to the same steady-state value, demonstrating robustness against large initial deviations. Different observer gains $(\hat{s}_{p1}, \hat{s}_{p2})$ mainly affect the convergence speed. Larger gains accelerate convergence, whereas excessively large gains introduce transient oscillations.}

\textcolor{black}{Fig.~\ref{gamma} shows the acceleration process from 200~rpm to 1000~rpm under full load. The estimation error remains bounded during the speed transition. Under different observer gains $\phi$, the steady-state performance after full-load application remains nearly identical, further confirming the stability and robustness of the proposed method.}

\textcolor{black}{Moreover, Table \ref{Table.THD} summarizes the current total harmonic distortion (THD) of the four methods at 1000 rpm under full-load conditions, together with their computation time in DSP implementation. The proposed method exhibits a slightly higher current THD but requires shorter computation time, indicating reduced computational resource requirements.}

\section{Discussion}

\textcolor{black}{Experimental results demonstrate the high estimation accuracy and robustness of the proposed EFC method. Nevertheless, two aspects warrant further investigation. First, extending the method to interior PMSMs (IPMSMs) requires incorporating an inductance matrix into the observer dynamics \cite{khlaief2011nonlinear}. While theoretically compatible due to the adaptive mechanism, this extension requires experimental validation. Second, although the EFC framework eliminates resistance and flux mismatches, its decoupling capability against inductance variations remains limited. Eliminating this remaining inductance sensitivity in future designs holds the potential to realize a fully model-free sensorless drive system.}

\begin{table}[!t] \color{black}
  \caption{Current THD and Computational Time of Four Methods}
  \renewcommand\arraystretch{1.25} 
  \small
  \centering
  \setlength{\tabcolsep}{2mm} 
  \begin{tabular}{c c c c c c} \hline \hline
     & Method 1 & Method 2 & Method 3 & Proposed EFC   \\ \hline 
    THD (\%) & $\underline{5.39}$ & $6.44$ & $5.55$ & ${6.03}$ \\
    Time ($\mu s$) & ${16.18}$ & $72.32$ &$ \underline{15.59}$ & ${17.01}$ \\
    \hline  \hline 
  \end{tabular}
  \label{Table.THD}
\end{table}

\section{Conclusion}

To address the estimation accuracy degradation caused by parameter mismatches in sensorless control, this paper proposes an EFC method based on a flux observer integrated with an equivalent flux update law. Formulated within a geometric error framework, the proposed scheme dynamically updates both the amplitude and angle of the equivalent flux based on the error dynamics under parameter mismatches. Furthermore, a rigorous stability proof for the overall observer system is established by leveraging Lyapunov stability theory. These analytical insights explicitly reveal the position estimation behavior under various mismatch scenarios, which are subsequently validated through comparative experiments. \textcolor{black}{Benchmarked against three methods, the proposed EFC strategy guarantees exceptional estimation robustness under diverse operating conditions; even under severe resistance and flux mismatches, the estimation error is effectively restored to its original baseline state, confirming its strong robustness to parameter variations.}

\appendices
\section{Proof of Lemma 1} \label{proof_Lemma1}
\begin{proof}
\textcolor{black}{By symmetry, it suffices to prove the inequality for $a \ge 0$, where it simplifies to $\tanh(a) \ge \frac{a}{1+a}$. This is equivalent to proving $1 - \tanh(a) \le \frac{1}{1+a}$. Using the exponential definition of $\tanh(a)$, we have:}
\begin{equation}
    \textcolor{black}{1 - \tanh(a) = \frac{2}{e^{2a}+1} \le \frac{1}{1+a}
\label{eq73}}
\end{equation}

\textcolor{black}{Rearranging the inequality in \eqref{eq73} yields:}
\begin{equation}
   \textcolor{black}{ e^{2a} \ge 1 + 2a}
\label{eq74}
\end{equation}
\textcolor{black}{which unconditionally holds according to the fundamental exponential inequality $e^x \ge 1+x, \forall x \in \mathbb{R}$. Equality holds if and only if $a=0$, which completes the proof.}

\end{proof}

\section{Proof of Lemma 2} \label{proof_Lemma2}
\begin{proof} \textcolor{black}{
The proofs for $Q_1$, $Q_1^*$, and $Q_2^*$ are presented below, respectively:}
\begin{equation}\textcolor{black}{
    \begin{aligned}
        Q_1 & = (x_1 - \hat{x}_1) \cos \theta + (x_2 - \hat{x}_2) \sin \theta \\
        & = \psi_m (\cos \theta - \cos \hat{\theta}) \cos \theta \\
        & \quad + \psi_m (\sin \theta - \sin \hat{\theta}) \sin \theta \\
        & = \psi_m (1 - \cos \tilde{\theta})
    \end{aligned}
\label{eqA81}}
\end{equation}
\begin{equation}\textcolor{black}{
    \begin{aligned}
        Q_1^* & = (x_1 - \hat{x}_1^*) \cos \theta + (x_2 - \hat{x}_2^*) \sin \theta \\
        & = (\Delta L_s i_\alpha + \psi_m \cos \theta - \psi_z \cos \hat{\theta}^*) \cos \theta \\
        & \quad + (\Delta L_s i_\beta + \psi_m \sin \theta - \psi_z \sin \hat{\theta}^*) \sin \theta \\
        & = \Delta L_s i_d + \psi_m - \psi_z \cos \tilde{\theta}^*
    \end{aligned}
\label{eqA82}}
\end{equation}
\begin{equation}\textcolor{black}{
    \begin{aligned}
        Q_2^* & = (x_1 - \hat{x}_1^*) i_\alpha + (x_2 - \hat{x}_2^*) i_\beta \\
        & = (\Delta L_s i_\alpha + \psi_m \cos \theta - \psi_z \cos \hat{\theta}^*) i_\alpha \\
        & \quad + (\Delta L_s i_\beta + \psi_m \sin \theta - \psi_z \sin \hat{\theta}^*) i_\beta \\
        & =  \Delta L_s \|i_{\alpha\beta}\|^2 + \psi_m i_d - \psi_z \hat{i}_d^*
    \end{aligned}
\label{eqA83}}
\end{equation}

\end{proof}

\section{Proof of Theorem 1} \label{proof_1}
\begin{proof}
    
    \textcolor{black}{The Lyapunov function \eqref{eq39} for the flux observer under accurate parameters takes the following form:}
    \begin{equation}\textcolor{black}{
        \begin{aligned}
            \dot V_1  &  =  \phi \xi_1(\tilde{x}) \tilde{x}^\top \left( \tilde{x} - \psi_m \begin{bmatrix}
        \cos \theta \\
        \sin \theta
    \end{bmatrix} \right) \\
            & = \phi \left( 2\psi_m Q_1 - \|\tilde{x}\|^2 \right) \left( \|\tilde{x}\|^2 - \psi_m Q_1 \right) \\
        \end{aligned}
    \label{eq40}}
    \end{equation}

\textcolor{black}{Equation \eqref{eq40} clearly describes a quadratic polynomial in terms of $\|\tilde{x}\|^2$ with a negative leading coefficient. Because its discriminant $\mathcal{D} = 4 \phi^2 \psi_m^2 Q_1^2$ is strictly positive, the equation possesses two distinct real roots. According to the Lyapunov stability criteria for such dynamics, $\|\tilde{x}\|^2$ will converge to the larger root, yielding $\lim\limits_{t\to\infty} \|\tilde{x}\| = \sqrt{2 \psi_m Q_1} \leq 2\psi_m$.}

\textcolor{black}{It should be noted that according to Lyapunov stability theory, $\|\tilde{x}\|^2$ would converge to zero if the initial state satisfies $V_1(0) < \psi_m Q_1$. However, since the observer is typically initialized with large values in practical applications, this scenario is omitted here and in the subsequent analysis of $V_2$.}

\end{proof}

\section{Proof of Theorem 2} \label{proof_2}
\begin{proof}
    
    \textcolor{black}{According to \eqref{eq17}-\eqref{eq19} and {Lemma \ref{lemma 2}}, the time derivative of \eqref{eq45} is described as:}
    \begin{equation}\textcolor{black}{
        \begin{aligned}
            \dot V_2  & =  \tilde{x}^{*\top} \dot{\tilde{x}}^* \\
            & = \phi \tilde{x}^{*\top}\Big( \tilde{x}^* - \psi_m \begin{bmatrix} \cos \theta \\ \sin \theta  \end{bmatrix}  - \Delta L_s i_{\alpha \beta } \Big) \left( \xi_1^* + \xi_2^* \right) \\
            & \quad -\Delta R_s \tilde{x}^{*\top} i_{\alpha \beta } \\
            & = \phi \big( \|\tilde{x}^*\|^2 - \psi_m Q_1^* - \Delta L_s Q_2^* \big) \big( -\|\tilde{x}^*\|^2   \\
            & \quad + 2 \psi_m Q_1^* + \xi_2^* \big) - \Delta R_s Q_2^*
        \end{aligned}
    \label{eq47}}
    \end{equation}

    \textcolor{black}{Next, the impacts of mismatches in flux, resistance, and inductance on the stability of the nonlinear flux observer are analyzed separately. Note that the effect of the discriminant $\mathcal{D}^*$ in \eqref{eq47} does not need explicit consideration. Since the leading coefficient is negative, $\mathcal{D}^* \leq 0$ directly implies $\dot V_2 \leq 0$, ensuring stability and convergence of $\|\tilde{x}^*\|$ to zero. Therefore, the subsequent analysis assumes $\mathcal{D}^* > 0$.}

    \subsection{Effect of Flux Mismatch}
   \textcolor{black}{ When exclusively considering the effect of flux mismatch $\Delta \psi_m < \psi_m$, the resistance and inductance mismatches are set to zero, i.e., $\Delta R_s = \Delta L_s = 0$. In this case, one obtains $Q_1^* = \psi_m - \psi_z \cos \tilde{\theta}^*$, $Q_2^* = \psi_m i_d - \psi_z \hat{i}_d^*$, and $\xi_2^2 = \psi_z^2 - \psi_m^2$. Substituting these expressions into \eqref{eq47}, the equation is rewritten as follows:}
 \begin{equation} \textcolor{black}{
     \begin{aligned}
         \dot V_2 & = \phi \big( \|\tilde{x}^*\|^2 - r_{\psi_{m1}} \big) \big( -\|\tilde{x}^*\|^2 + r_{\psi_{m2}} \big)
     \end{aligned}
     \label{eq101}}
 \end{equation}
\textcolor{black}{where $r_{\psi_{m1}} = \psi_m^2 - \psi_m \psi_z \cos \tilde{\theta}^*$ and $r_{\psi_{m2}} = \psi_m^2 + \psi_z^2 - 2 \psi_m \psi_z \cos \tilde{\theta}^*$ denote the two roots of $\dot V_2$ under flux mismatch. According to Lyapunov stability theory, $\|\tilde{x}^*\|^2$ converges to the larger root. Therefore, when $\psi_z \geq \psi_m \cos \tilde{\theta}^*$, one has $\lim\limits_{t\to\infty} \|\tilde{x}^*\|^2 = r_{\psi_{m2}}$, whereas when $0 < \psi_z < \psi_m \cos \tilde{\theta}^*$, it follows that $\lim\limits_{t\to\infty} \|\tilde{x}^*\| ^2= r_{\psi_{m1}}$.}

\textcolor{black}{Fig. \ref{V_2_psi} illustrates the steady-state error boundaries $r_{\psi_{m1}}$ and $r_{\psi_{m2}}$ under flux mismatch. The results indicate that the estimation error boundaries are directly determined by the deviation of $\psi_z$ from $\psi_m$, where a larger mismatch leads to an increased steady-state error.}

\begin{figure}[!t]
    \centering
    \setlength{\subfigbottomskip}{-2pt} %
    \setlength{\subfiglabelskip}{-10pt} %
    \setlength{\subfigcapskip}{-6pt}   %
    \subfigure[]{
        \includegraphics[width=0.4\linewidth]{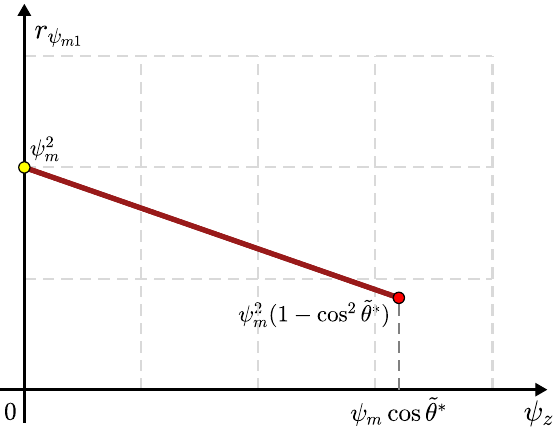}
    }
    %
    \subfigure[]{
        \includegraphics[width=0.4\linewidth]{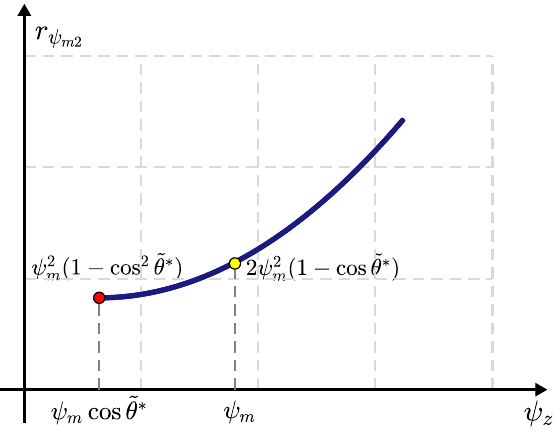}
    }
    \caption{\textcolor{black}{Steady-state estimation error boundaries under flux mismatch. (a) $0 < \psi_z < \psi_m \cos \tilde{\theta}^*$; (b) $\psi_z \geq \psi_m \cos \tilde{\theta}^*$.}}
    \label{V_2_psi}
\end{figure}

\subsection{Effect of Inductance Mismatch}
    \textcolor{black}{When examining the effect of inductance mismatch $\Delta L_s < L_s$, the flux and resistance mismatches are set to zero, that is, $\Delta \psi_m = \Delta R_s = 0$. In this case, the corresponding terms become $Q_1^* = \Delta L_s i_d + \psi_m - \psi_m \cos \tilde{\theta}^*$, $Q_2^* = \Delta L_s \|i_{\alpha\beta}\|^2 + \psi_m i_d - \psi_m \hat{i}^*_d$, and $\xi_2^* = \Delta L_s^2 \|i_{\alpha\beta}\|^2 - 2 \Delta L_s \psi_m \hat{i}_d^*$. Accordingly, \eqref{eq47} is rewritten as follows:}
\begin{equation}\textcolor{black}{
     \begin{aligned}
         \dot V_2 & = \phi \big( \|\tilde{x}^*\|^2 - r_{L_{s1}} \big) \big( -\|\tilde{x}^*\|^2 + r_{L_{s2}} \big)
     \end{aligned}
     \label{eq102}}
 \end{equation}
\textcolor{black}{where $r_{L_{s1}} = \psi_m^2(1 - \cos \tilde{\theta}^*) + \Delta L_s^2 \|i_{\alpha\beta}\|^2 + \Delta L_s \psi_m (2 i_d - \hat{i}_d^*)$  and $r_{L_{s2}} = 2\psi_m^2(1 - \cos \tilde{\theta}^*) + \Delta L_s^2 \|i_{\alpha\beta}\|^2 + 2 \Delta L_s \psi_m (i_d - \hat{i}_d^*) $ denote the two roots of $\dot V_2$ under inductance mismatch. Similarly, $\|\tilde{x}^*\|^2$ converges to the larger of the two roots. It thus remains to identify which root is dominant:}
\begin{equation}\textcolor{black}{
    r_{L_{s2}} - r_{L_{s1}} = \psi_m \big( \psi_m (1-\cos\tilde{\theta}^*) - \Delta L_s \hat{i}_d^* \big)
    \label{eq103}}
\end{equation}

\textcolor{black}{It is clear that \eqref{eq103} is positive, since $\psi_m$ and $(1 - \cos \tilde{\theta}^*)$ are smaller in magnitude than $\Delta L_s$, and $\hat{i}_d^*$ is of comparable magnitude. Therefore, it follows that $\lim\limits_{t\to\infty} \|\tilde{x}^*\| ^2= r_{L_{s2}}$ and:}
\begin{equation}\textcolor{black}{
    \Delta L_{s,a} = - \frac{\psi_m (i_d - \hat{i}_d^*)}{\|i_{\alpha\beta}\|^2}
    \label{eq104}}
\end{equation}
\textcolor{black}{where $\Delta L_{s,a}$ denotes the axis of symmetry of \(r_{L_{s2}}\) as a function of \(\Delta L_s\).}

\textcolor{black}{Fig. \ref{V_2_Ls} illustrates the steady-state error boundaries $r_{L_{s2}}$ under inductance mismatch. In both cases, $r_{L_{s2}}$ follows a parabolic trend, reaching a minimum at the symmetry axis $\Delta L_{s,a}$. Specifically, a properly sized negative mismatch reduces the error when $i_d > \hat{i}_d^*$, whereas a positive mismatch is beneficial when $i_d < \hat{i}_d^*$. However, excessive deviations in either direction significantly degrade system stability and may lead to observer divergence.}

\begin{figure}[!t]
    \centering
    \setlength{\subfigbottomskip}{-2pt} %
    \setlength{\subfiglabelskip}{-10pt} %
    \setlength{\subfigcapskip}{-6pt}   %
    \subfigure[]{
        \includegraphics[width=0.4\linewidth]{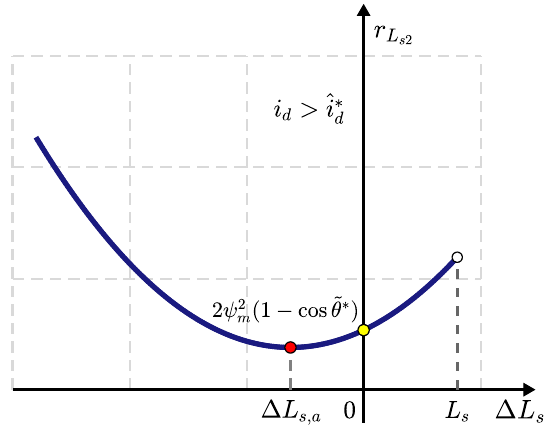}
    }
    %
    \subfigure[]{
        \includegraphics[width=0.4\linewidth]{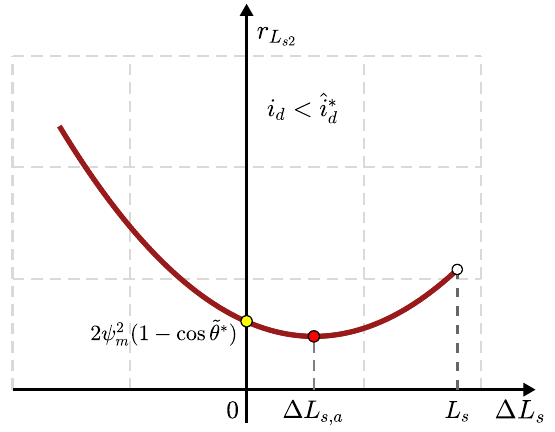}
    }
    \caption{\textcolor{black}{Steady-state estimation error boundaries under inductance mismatch. (a) $i_d > \hat{i}_d^*$; (b) $i_d < \hat{i}_d^*$.}}
    \label{V_2_Ls}
\end{figure}

\subsection{Effect of Resistance Mismatch}
\textcolor{black}{When analyzing the effect of resistance mismatch $\Delta R_s < R_s$, the flux and inductance mismatches are set to zero, namely $\Delta \psi_m = \Delta L_s = 0$. Under this condition, the corresponding terms reduce to $Q_1^* = \psi_m - \psi_m \cos \tilde{\theta}^*$, $Q_2^* = \psi_m i_d - \psi_m \hat{i}_d^*$, and $\xi_2^* = 0$. Thus, \eqref{eq47} is rewritten as follows:}
\begin{equation}\textcolor{black}{
     \begin{aligned}
         \dot V_2 & = \phi \big( \|\tilde{x}^*\|^2 - \psi_m Q_1^* \big) \big( -\|\tilde{x}^*\|^2  + 2\psi_m Q_1^* \big) \\
         & \quad - \Delta R_s \psi_m (i_d - \hat{i}_d^*) \\
         & = \phi \big( \|\tilde{x}^*\|^2 - r_{R_{s1}} \big) \big( -\|\tilde{x}^*\|^2 + r_{R_{s2}} \big)
     \end{aligned}
     \label{eq105}}
 \end{equation}
\textcolor{black}{where $r_{R_{s1}}$ and $r_{R_{s2}}$ denote the two roots of $\dot V_2$ under resistance mismatch. In the nominal case, the roots are $\psi_m Q_1^*$ and $2\psi_m Q_1^*$. The resistance mismatch introduces a constant perturbation that vertically shifts the quadratic, yielding the modified roots $r_{R_{s1}}$ and $r_{R_{s2}}$. Consequently, it follows that $\lim\limits_{t\to\infty} \|\tilde{x}^*\| ^2= r_{R_{s2}}$:}
\begin{equation}\textcolor{black}{
    r_{R_{s2}} = \frac{3\psi_m Q_1^{*}}{2} + \sqrt{\frac{\psi_m^2 Q_1^{*2}}{4} - \dfrac{\Delta R_s \psi_m (i_d - \hat{i}_d^*)}{\phi}}
    \label{eq106}}
\end{equation}

\textcolor{black}{Fig. \ref{V_2_Rs} illustrates the steady-state error boundaries $r_{R_{s2}}$ under resistance mismatch. Based on \eqref{eq106}, when $i_d < \hat{i}_d^*$, a negative mismatch reduces the convergence radius, whereas it amplifies the error when $i_d > \hat{i}_d^*$. However, excessive deviations lead to either a significantly increased error or a negative term under the radical in \eqref{eq106}, the latter resulting in the absence of a real solution; both of which potentially lead to observer divergence.}

\begin{figure}[!t]
    \centering
    \setlength{\subfigbottomskip}{-2pt} %
    \setlength{\subfiglabelskip}{-10pt} %
    \setlength{\subfigcapskip}{-6pt}   %
    \subfigure[]{
        \includegraphics[width=0.4\linewidth]{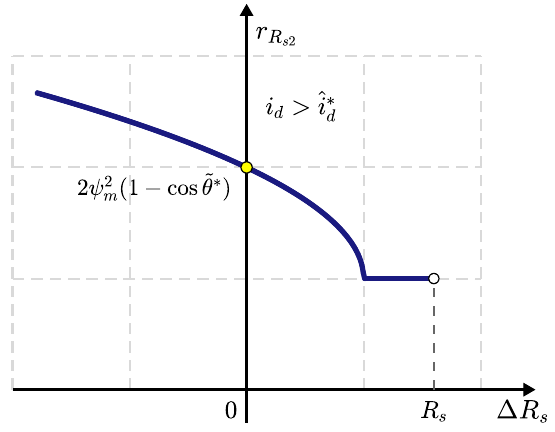}
    }
    %
    \subfigure[]{
        \includegraphics[width=0.4\linewidth]{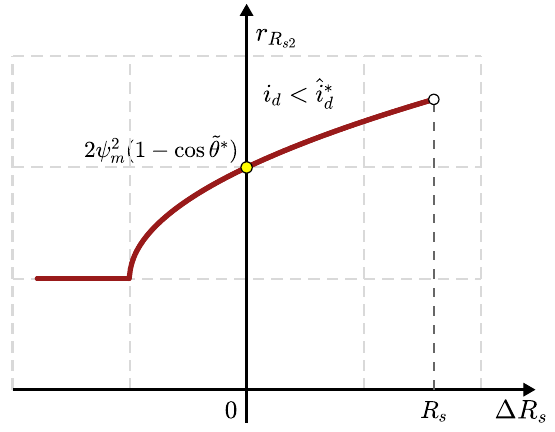}
    }
    \caption{\textcolor{black}{Steady-state estimation error boundaries under resistance mismatch. (a) $i_d > \hat{i}_d^*$; (b) $i_d < \hat{i}_d^*$.}}
    \label{V_2_Rs}
\end{figure}
    
\end{proof}

\section{Proof of Theorem 3} \label{proof_3}
\begin{proof}
    The time derivative of $\tilde{\psi}_z$ is given by $\dot{\tilde{\psi}}_z = - \hat{s}_{p1} \tanh (\hat{s}_{p2} \tilde{\psi}_z) - \| \dot{\eta} (\hat{x}^*)\|$. Considering the upper bound $\| \dot{\eta} (\hat{x}^*)\| \leq \Psi$, the Lyapunov function $V_3$ satisfies:
    \begin{equation}
        \dot V_3 = \tilde{\psi}_z \dot{\tilde{\psi}}_z \leq - \hat{s}_{p1} |\tilde{\psi}_z| \cdot |\tanh (\hat{s}_{p2} \tilde{\psi}_z)| +  |\tilde{\psi}_z| \Psi
    \label{eq61}
    \end{equation}

    According to {Lemma \ref{lemma 1}}, the inequality in \eqref{eq61} can be further bounded by:
    \begin{equation}
        \dot V_3 \leq -\frac{\hat{s}_{p1} \hat{s}_{p2}}{1 + \hat{s}_{p2} |\tilde{\psi}_z|} |\tilde{\psi}_z|^2 +  |\tilde{\psi}_z| \Psi
    \label{eq62}
    \end{equation}

    Letting $\gamma = \sqrt{2V_3} = |\tilde{\psi}_z|$, it follows that $\dot{V}_3 = \gamma \dot{\gamma}$. Thus, \eqref{eq62} becomes:
    \begin{equation}
        \dot \gamma \leq -\frac{\hat{s}_{p1} \hat{s}_{p2} \gamma }{1+\hat{s}_{p2} \gamma } + \Psi = f(\gamma)
    \label{eq63}
    \end{equation}

    Evidently, $f(\gamma)$ is Lipschitz continuous and strictly decreasing for $\gamma \geq 0$, since $|f'(\gamma)| = \left| - \hat{s}_{p1} \hat{s}_{p2} / (1 + \hat{s}_{p2} \gamma)^2 \right| \leq \hat{s}_{p1} \hat{s}_{p2} $. Therefore, setting $f(\gamma) = 0$ yields a unique positive equilibrium point $\gamma^*$:
    \begin{equation}
        \gamma^* = \frac{\Psi}{\hat{s}_{p2} (\hat{s}_{p1} - \Psi)}
    \label{eq64}
    \end{equation}

    Based on the comparison theorem, $\gamma^*$ is a stable equilibrium, and all solutions of $\gamma$ will eventually converge to the region bounded by $\gamma^*$. Consequently, the transient bound and the ultimate steady-state bound are established as:
    \begin{equation}
        | \tilde{\psi }_z |_\infty  = \lim_{t \rightarrow  \infty} \gamma = \frac{\Psi }{\hat{s}_{p2} (\hat{s}_{p1} - \Psi) } 
    \label{eq68}
    \end{equation}
    where $\tilde{\psi}_z(0)$ denotes the initial value.

\end{proof}

\bibliographystyle{IEEEtran}
\bibliography{IEEEabrv, ref}

\begin{thebibliography}{10}
\providecommand{\url}[1]{#1}
\csname url@samestyle\endcsname
\providecommand{\newblock}{\relax}
\providecommand{\bibinfo}[2]{#2}
\providecommand{\BIBentrySTDinterwordspacing}{\spaceskip=0pt\relax}
\providecommand{\BIBentryALTinterwordstretchfactor}{4}
\providecommand{\BIBentryALTinterwordspacing}{\spaceskip=\fontdimen2\font plus
\BIBentryALTinterwordstretchfactor\fontdimen3\font minus \fontdimen4\font\relax}
\providecommand{\BIBforeignlanguage}[2]{{%
\expandafter\ifx\csname l@#1\endcsname\relax
\typeout{** WARNING: IEEEtran.bst: No hyphenation pattern has been}%
\typeout{** loaded for the language `#1'. Using the pattern for}%
\typeout{** the default language instead.}%
\else
\language=\csname l@#1\endcsname
\fi
#2}}
\providecommand{\BIBdecl}{\relax}
\BIBdecl

\bibitem{liu2016research}
X.~Liu, H.~Chen, J.~Zhao, and A.~Belahcen, ``Research on the performances and parameters of interior pmsm used for electric vehicles,'' \emph{IEEE Transactions on Industrial Electronics}, vol.~63, no.~6, pp. 3533--3545, 2016.

\bibitem{wang2019position}
G.~Wang, M.~Valla, and J.~Solsona, ``Position sensorless permanent magnet synchronous machine drives—a review,'' \emph{IEEE Transactions on Industrial Electronics}, vol.~67, no.~7, pp. 5830--5842, 2019.

\bibitem{bolognani2014design}
S.~Bolognani, S.~Calligaro, and R.~Petrella, ``Design issues and estimation errors analysis of back-emf-based position and speed observer for spm synchronous motors,'' \emph{IEEE Journal of Emerging and Selected Topics in Power Electronics}, vol.~2, no.~2, pp. 159--170, 2014.

\bibitem{he2023optimization}
W.~He, X.~Wu, and J.~Chen, ``Optimization design of pmsm sensorless control using generalized integrator,'' \emph{IEEE Transactions on Industrial Electronics}, vol.~71, no.~8, pp. 8625--8634, 2023.

\bibitem{zhang2023commutation}
H.~Zhang, L.~Deng, H.~Li, S.~Zheng, H.~Jin, and B.~Chen, ``Commutation point optimization method for sensorless bldc motor control using vector phase difference of back emf and current,'' \emph{IEEE/ASME Transactions on Mechatronics}, vol.~29, no.~1, pp. 423--433, 2023.

\bibitem{prabhakaran2020electromagnetic}
K.~Prabhakaran and A.~Karthikeyan, ``Electromagnetic torque-based model reference adaptive system speed estimator for sensorless surface mount permanent magnet synchronous motor drive,'' \emph{IEEE Transactions on Industrial Electronics}, vol.~67, no.~7, pp. 5936--5947, 2020.

\bibitem{yan2022mras}
X.~Yan and M.~Cheng, ``An mras observer-based speed sensorless control method for dual-cage rotor brushless doubly fed induction generator,'' \emph{IEEE Transactions on Power Electronics}, vol.~37, no.~10, pp. 12\,705--12\,714, 2022.

\bibitem{verrelli2019speed}
C.~M. Verrelli, S.~Bifaretti, E.~Carfagna, A.~Lidozzi, L.~Solero, F.~Crescimbini, and M.~Di~Benedetto, ``Speed sensor fault tolerant pmsm machines: From position-sensorless to sensorless control,'' \emph{IEEE Transactions on Industry Applications}, vol.~55, no.~4, pp. 3946--3954, 2019.

\bibitem{xiang2025sensorless}
F.~Xiang, K.~Yu, C.~Chen, and S.~Li, ``Sensorless control of pmsm drives using octant newton-raphson method and reduced-order ekf considering speed reversal and noise,'' \emph{IEEE Transactions on Power Electronics}, 2025.

\bibitem{wu2024sensorless}
X.~Wu, C.~Li, Y.~Zhang, S.~Chen, Z.~Ma, Y.~Han, X.~Zhang, and G.~Tan, ``Sensorless control of ipmsm equipped with lc sinusoidal filter based on full-order sliding mode observer and feedforward qpll,'' \emph{IEEE Transactions on Power Electronics}, vol.~39, no.~7, pp. 8072--8085, 2024.

\bibitem{yang2024rotor}
Q.~Yang, K.~Mao, S.~Zheng, and Y.~Le, ``Rotor position estimation based on fast terminal sliding mode for magnetic suspension centrifugal compressor drives,'' \emph{IEEE Transactions on Instrumentation and Measurement}, 2024.

\bibitem{lee2009sensorless}
J.~Lee, J.~Hong, K.~Nam, R.~Ortega, L.~Praly, and A.~Astolfi, ``Sensorless control of surface-mount permanent-magnet synchronous motors based on a nonlinear observer,'' \emph{IEEE Transactions on power electronics}, vol.~25, no.~2, pp. 290--297, 2009.

\bibitem{11099531}
F.~Zhou, X.~Wang, Z.~Yin, Y.~Shen, Y.~Liang, and H.~Zhao, ``Optimized robust observer-based encoderless control for spmsm,'' \emph{IEEE Transactions on Industrial Electronics}, vol.~72, no.~12, pp. 12\,632--12\,643, 2025.

\bibitem{ortega2010estimation}
R.~Ortega, L.~Praly, A.~Astolfi, J.~Lee, and K.~Nam, ``Estimation of rotor position and speed of permanent magnet synchronous motors with guaranteed stability,'' \emph{IEEE Transactions on Control Systems Technology}, vol.~19, no.~3, pp. 601--614, 2010.

\bibitem{yin2026nonlinear}
Z.~Yin, F.~Zhou, X.~Wang, Y.~Shen, J.~Liang, X.~Su, and H.~Zhao, ``Nonlinear bounded error compensation for flux based encoderless controller of pmsms,'' \emph{IEEE Transactions on Transportation Electrification}, 2026.

\bibitem{khlaief2011nonlinear}
A.~Khlaief, M.~Bendjedia, M.~Boussak, and M.~Gossa, ``A nonlinear observer for high-performance sensorless speed control of ipmsm drive,'' \emph{IEEE Transactions on Power Electronics}, vol.~27, no.~6, pp. 3028--3040, 2011.

\bibitem{xu2018improved}
W.~Xu, Y.~Jiang, C.~Mu, and F.~Blaabjerg, ``Improved nonlinear flux observer-based second-order soifo for pmsm sensorless control,'' \emph{IEEE Transactions on Power Electronics}, vol.~34, no.~1, pp. 565--579, 2018.

\bibitem{zhou2023robust}
S.~Zhou, Z.~Zhang, H.~Li, Z.~Li, Q.~Xing, X.~Liu, F.~Wang, and J.~Rodriguez, ``A robust encoderless control for pmsm drives: A revised hybrid active flux-based technique,'' \emph{IEEE Transactions on Power Electronics}, vol.~38, no.~11, pp. 14\,438--14\,449, 2023.

\bibitem{bernard2020estimation}
P.~Bernard and L.~Praly, ``Estimation of position and resistance of a sensorless pmsm: A nonlinear luenberger approach for a nonobservable system,'' \emph{IEEE Transactions on Automatic Control}, vol.~66, no.~2, pp. 481--496, 2020.

\bibitem{hinkkanen2011combined}
M.~Hinkkanen, T.~Tuovinen, L.~Harnefors, and J.~Luomi, ``A combined position and stator-resistance observer for salient pmsm drives: Design and stability analysis,'' \emph{IEEE Transactions on Power Electronics}, vol.~27, no.~2, pp. 601--609, 2011.

\bibitem{li2015position}
W.~Li, J.~Fang, H.~Li, and J.~Tang, ``Position sensorless control without phase shifter for high-speed bldc motors with low inductance and nonideal back emf,'' \emph{IEEE Transactions on Power Electronics}, vol.~31, no.~2, pp. 1354--1366, 2015.

\bibitem{li2018sensorless}
H.~Li, Z.~Wang, C.~Wen, and X.~Wang, ``Sensorless control of surface-mounted permanent magnet synchronous motor drives using nonlinear optimization,'' \emph{IEEE Transactions on Power Electronics}, vol.~34, no.~9, pp. 8930--8943, 2018.

\bibitem{bernard2018convergence}
P.~Bernard and L.~Praly, ``Convergence of gradient observer for rotor position and magnet flux estimation of permanent magnet synchronous motors,'' \emph{Automatica}, vol.~94, pp. 88--93, 2018.

\bibitem{islam2013sensorless}
M.~N. Islam and R.~J. Seethaler, ``Sensorless position control for piezoelectric actuators using a hybrid position observer,'' \emph{IEEE/ASME Transactions On Mechatronics}, vol.~19, no.~2, pp. 667--675, 2013.

\bibitem{choi2016robust}
J.~Choi, K.~Nam, A.~A. Bobtsov, A.~Pyrkin, and R.~Ortega, ``Robust adaptive sensorless control for permanent-magnet synchronous motors,'' \emph{IEEE Transactions on Power Electronics}, vol.~32, no.~5, pp. 3989--3997, 2016.

\bibitem{li2025sensorless}
B.~Li, J.~Zou, Y.~Xu, S.~Li, and S.~Jin, ``Sensorless control of dtp-synrm with hybrid flux observer and disturbance observer considering magnetic saturation and cross-coupling effect,'' \emph{IEEE Transactions on Power Electronics}, 2025.

\bibitem{kivanc2018sensorless}
O.~C. Kivanc and S.~B. Ozturk, ``Sensorless pmsm drive based on stator feedforward voltage estimation improved with mras multiparameter estimation,'' \emph{IEEE/ASME Transactions on Mechatronics}, vol.~23, no.~3, pp. 1326--1337, 2018.

\bibitem{hamida2012adaptive}
M.~A. Hamida, J.~De~Leon, A.~Glumineau, and R.~Boisliveau, ``An adaptive interconnected observer for sensorless control of pm synchronous motors with online parameter identification,'' \emph{IEEE Transactions on Industrial Electronics}, vol.~60, no.~2, pp. 739--748, 2012.

\bibitem{liu2022second}
Z.~Liu, J.~Nie, H.~Wei, L.~Chen, F.~Wu, and M.~Lv, ``Second-order eso-based current sensor fault-tolerant strategy for sensorless control of pmsm with b-phase current,'' \emph{IEEE/ASME Transactions on Mechatronics}, vol.~27, no.~6, pp. 5427--5438, 2022.

\bibitem{nguyen2013modeling}
T.~D. Nguyen, G.~Foo, K.~Tseng, and D.~M. Vilathgamuwa, ``Modeling and sensorless direct torque and flux control of a dual-airgap axial flux permanent-magnet machine with field-weakening operation,'' \emph{IEEE/ASME Transactions On Mechatronics}, vol.~19, no.~2, pp. 412--422, 2013.

\bibitem{wang2024eso}
K.~Wang, G.~Wang, G.~Zhang, Q.~Wang, B.~Li, and D.~Xu, ``Eso-based robust hybrid flux observer with active flux error estimation for position sensorless pma-synrm drives,'' \emph{IEEE Transactions on Transportation Electrification}, vol.~11, no.~1, pp. 2049--2060, 2024.

\bibitem{woldegiorgis2022sensorless}
A.~T. Woldegiorgis, X.~Ge, Y.~Zuo, H.~Wang, and M.~Hassan, ``Sensorless control of interior permanent magnet synchronous motor drives considering resistance and permanent magnet flux linkage variation,'' \emph{IEEE Transactions on Industrial Electronics}, vol.~70, no.~8, pp. 7716--7730, 2022.

\bibitem{11025165}
D.~Yang, S.~Huang, W.~Liao, X.~Wu, X.~Yu, T.~Wu, D.~Luo, and S.~Huang, ``Iladrc-based sensorless ipmsm strategy with adaptive harmonic filtering-extended state observer for current quality improvement,'' \emph{IEEE Transactions on Power Electronics}, vol.~41, no.~2, pp. 1752--1763, 2026.

\bibitem{10949742}
S.~Xu, A.~Shen, Q.~Tang, M.~Wang, P.~Luo, X.~Luo, and J.~Xu, ``Surface-pmsm sensorless control strategy based on reduced-order linear kalman filter cooperating with prediction error rolling compensation,'' \emph{IEEE Transactions on Power Electronics}, vol.~40, no.~8, pp. 10\,804--10\,815, 2025.

\bibitem{zhou2023robust_1}
S.~Zhou, Z.~Zhang, H.~Li, Z.~Li, Q.~Xing, X.~Liu, F.~Wang, and J.~Rodriguez, ``A robust encoderless control for pmsm drives: A revised hybrid active flux-based technique,'' \emph{IEEE Transactions on Power Electronics}, vol.~38, no.~11, pp. 14\,438--14\,449, 2023.

\end{thebibliography}

\vspace{-1.65cm} 
\begin{IEEEbiography}[{\includegraphics[width=1in,height=1.25in,clip,keepaspectratio]{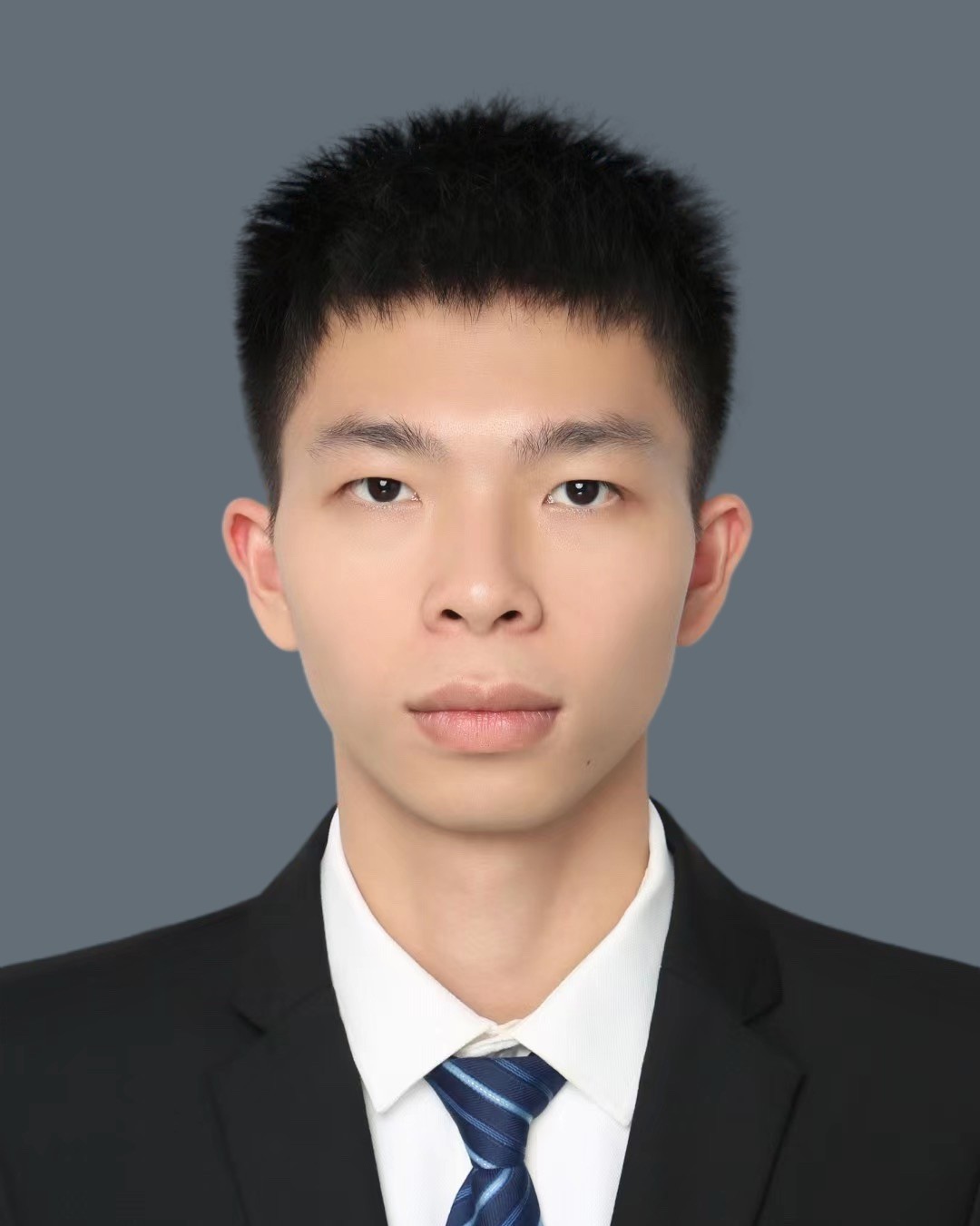}}]
{Fobao Zhou} (Graduate Student Member, IEEE) received the B.Eng. degree in Robot Engineering from Guangzhou University, Guangzhou, China, in 2022. He is currently pursuing the Ph.D. degree at The Hong Kong University of Science and Technology (Guangzhou), Guangzhou, China. His research interests include the integration of machine learning and control theory for the development of advanced intelligent systems. He aims to contribute to the field of engineering through innovative research and practical applications in these areas.
\end{IEEEbiography}

\vspace{-1.45cm} 
\begin{IEEEbiography}[{\includegraphics[width=1in,height=1.25in,clip,keepaspectratio]{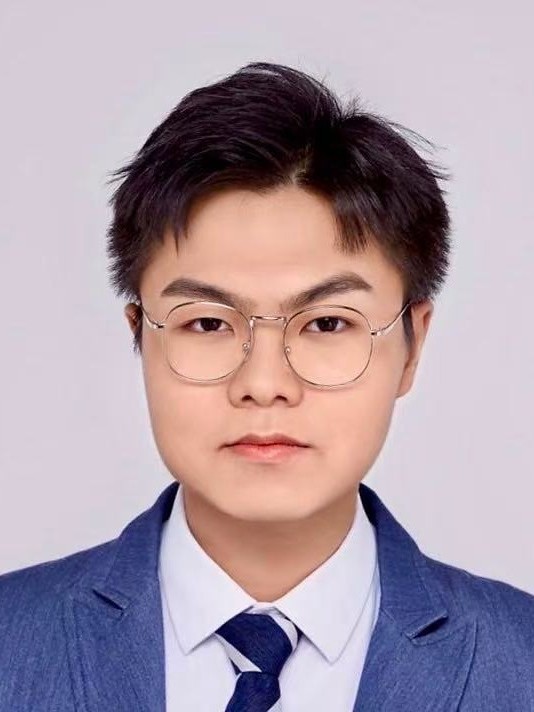}}]
{Jiaqiao Liang} received the B.Eng. degree in Mechanical Design, Manufacturing, and Automation from Guangzhou University, Guangzhou, China, in 2023. He is currently pursuing the M.Eng. degree in Mechanical Engineering at South China University of Technology, Guangzhou, China. His research focuses on nonlinear and robust control, underwater robotic systems, and bioinspired actuation, with emphasis on modeling and control of bistable-actuated soft robots and cross-medium locomotion.

\end{IEEEbiography}


\vspace{-1.45cm}
\begin{IEEEbiography}[{\includegraphics[width=1in,height=1.25in,clip,keepaspectratio]{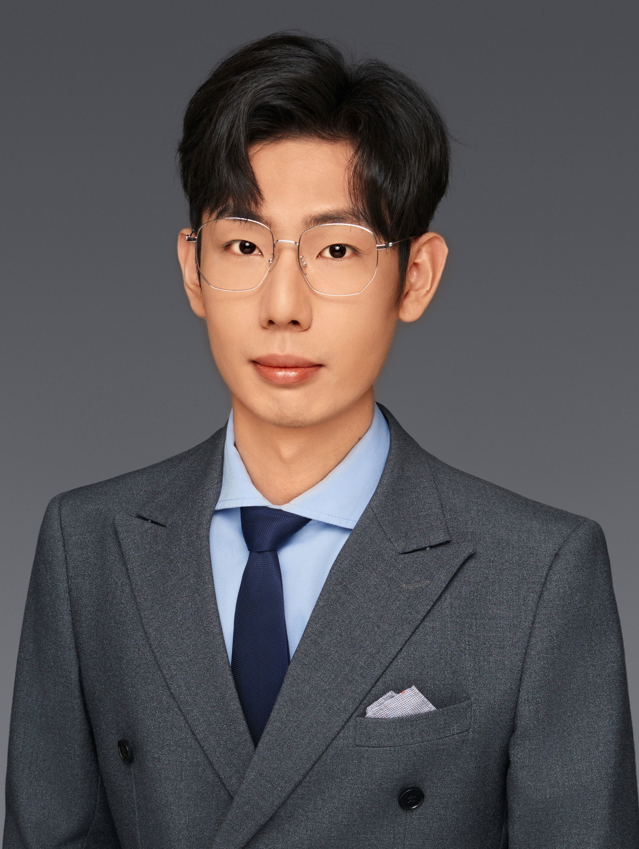}}]
{Zhenxiao Yin} (Graduate Student Member, IEEE) received a B.Sc. degree in mechanical engineering from Paderborn University (UPB), Paderborn, Germany, in 2018. Subsequently, he earned double M.Sc. degrees in Mechanical Engineering, Mechatronics and Robotics from the Technical University of Munich (TUM), Munich, Germany, in 2021 and 2022, where he was the recipient of the German National Scholarship. He is pursuing a Ph.D. in Robotics and Autonomous Systems at the Hong Kong University of Science and Technology (HKUST), while he is studying in the department of Mechanical and Aerospace Engineering at HKUST, Guangzhou and Hong Kong SAR, China. His research interests encompass learning-based control and state drop compensator for electric machines in flying vehicles.
\end{IEEEbiography}


\vspace{-1.3cm}
\begin{IEEEbiography}[{\includegraphics[width=1in,height=1.25in,clip,keepaspectratio]{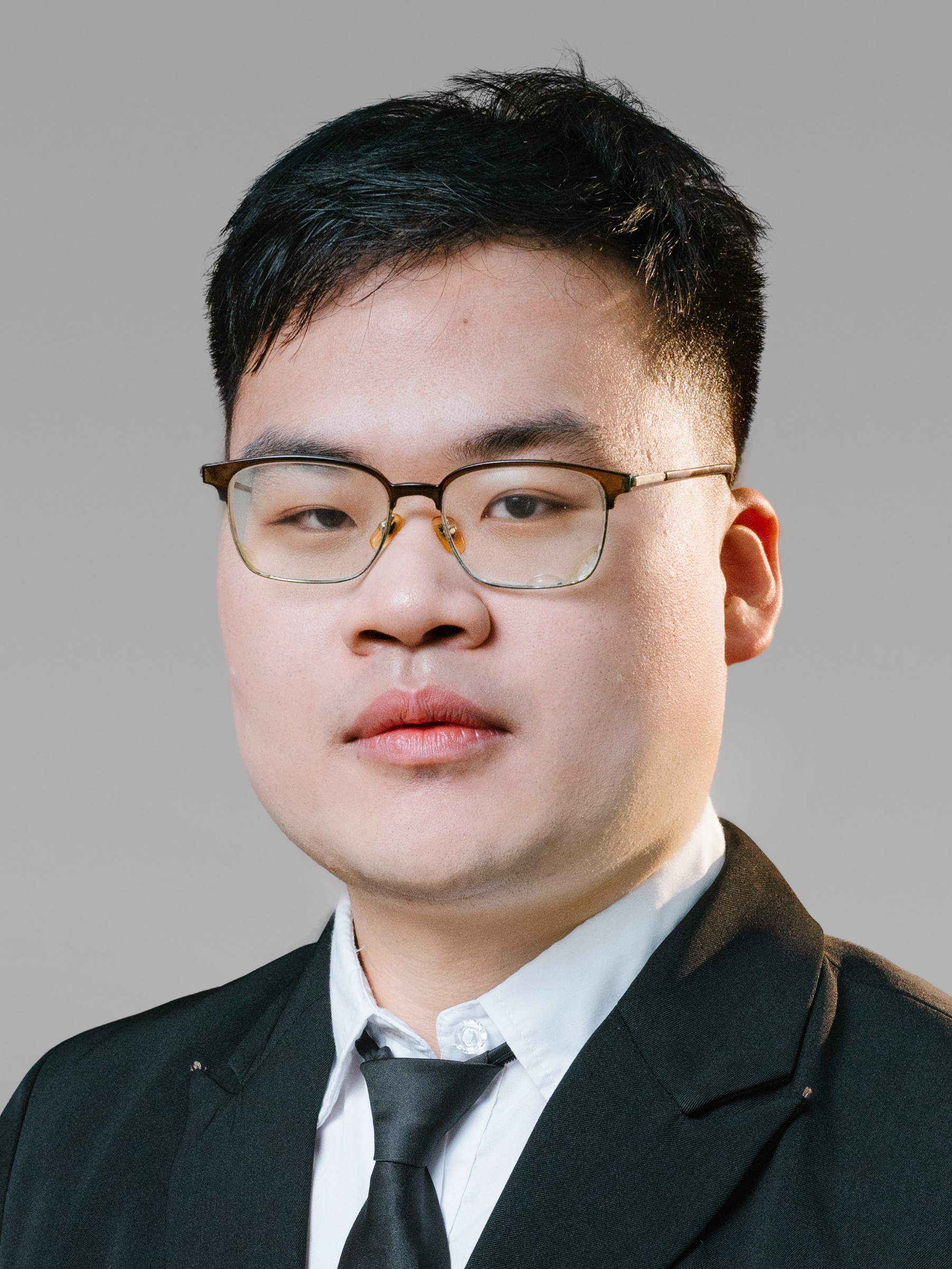}}]
{Xueyan Wang} received the B.S. degree in electrical engineering and automation and the M.S. degree in electrical engineering from Xi’an University of Science and Technology, Xi’an, China, in 2020 and 2023, respectively. He is currently a Research Assistant with The Hong Kong University of Science and Technology (Guangzhou), Guangzhou, China. His research interests include model predictive control of permanent magnet synchronous motor drives, power electronics, and power transmission.

\end{IEEEbiography}

\vspace{-1.45cm}
\begin{IEEEbiography}[{\includegraphics[width=1in,height=1.25in,clip,keepaspectratio]{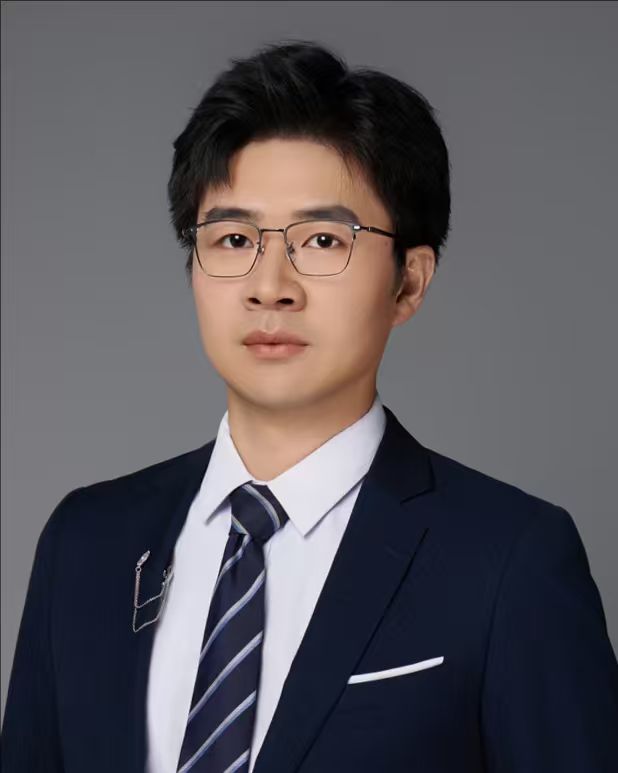}}]{Yang Shen} (Graduate Student Member, IEEE) received the B.Eng. and the M.Eng. degree from the North China University of Technology, Beijing, China, all in electrical engineering in 2018 and 2021, respectively. He is currently pursuing a Ph.D. in the Robotics and Autonomous Systems Thrust, Systems Hub, The Hong Kong University of Science and Technology (Guangzhou), Guangzhou, China. His research interests include electric machines and drives, robotic applications of motor drives, and electrified transportation systems.
 \end{IEEEbiography}

 \vspace{-1.25cm}
\begin{IEEEbiography}[{\includegraphics[width=1in,height=1.25in,clip,keepaspectratio]{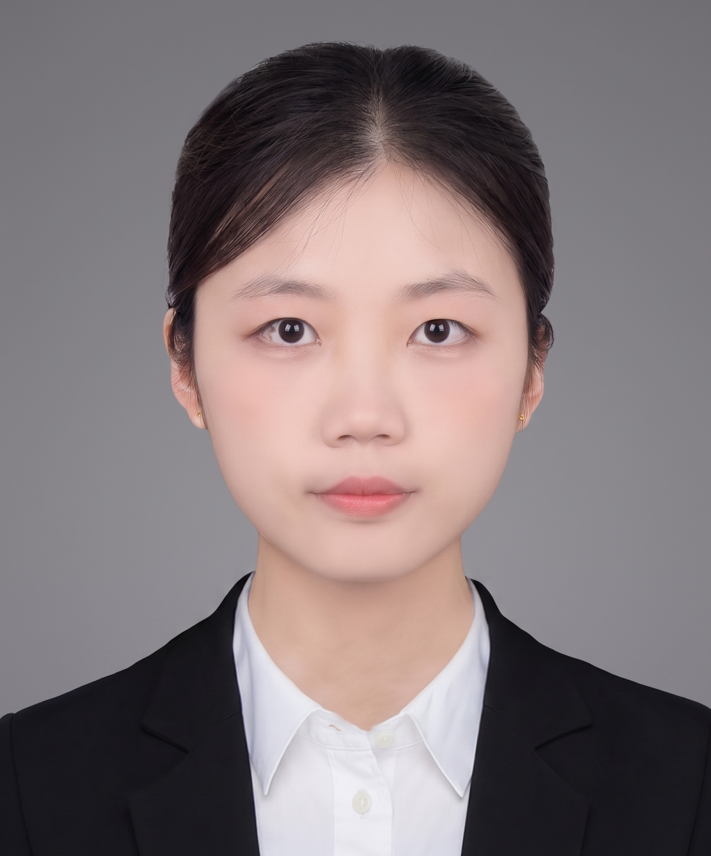}}]{Zhongyu Shi} is expected to receive the B.A. degree in Engineering Science from the University of Oxford in 2026. She is currently working toward the M.Eng. degree in Engineering Science at the University of Oxford. She is also a visiting student at the Robotics and Autonomous Systems Thrust, Systems Hub, The Hong Kong University of Science and Technology (Guangzhou). Her research interests include permanent magnet synchronous motor drives, parameter identification, and control.
 \end{IEEEbiography}

\vspace{-1.45cm}
\begin{IEEEbiography}[{\includegraphics[width=1in,height=1.25in,clip,keepaspectratio]{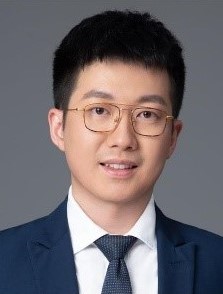}}]
{Hang Zhao} (Member, IEEE) received the B.Eng. and M.Eng. degrees from Huazhong University of Science and Technology, Wuhan, China, and the Ph.D. degree from City University of Hong Kong, Hong Kong SAR, China, all in electrical engineering in 2015, 2017, and 2021, respectively.
He is currently an Assistant Professor at the Robotics and Autonomous Systems Thrust, Systems Hub, The Hong Kong University of Science and Technology (Guangzhou), Guangzhou, China, and also an Affiliate Assistant Professor at the Department of Electronic and Computer Engineering, The Hong Kong University of Science and Technology, Hong Kong SAR, China. He was a Research Fellow at The University of Hong Kong, Hong Kong SAR, China, in 2021.
His research interests include electric machines and drives, applications of motor drives in robots, and electrified transportation.
\end{IEEEbiography}


\end{document}